%% file: main.tex
\documentclass[acmsmall]{acmart}
\AtBeginDocument{%
  }

\setcopyright{cc}
\setcctype{by}
\acmDOI{10.1145/3808284}
\acmYear{2026}
\acmJournal{PACMPL}
\acmVolume{10}
\acmNumber{PLDI}
\acmArticle{206}
\acmMonth{6}
\acmSubmissionID{pldi26main-p228-p}
\received{2025-11-12}
\received[accepted]{2026-04-03}

\usepackage[utf8]{inputenc}
\usepackage[T1]{fontenc}
\usepackage{microtype}
\usepackage{booktabs}       %% For formal tables
\usepackage{subcaption}     %% For complex figures with subfigures/subcaptions
\usepackage{braket}
\usepackage{proof}
\usepackage{pgfplots}
\pgfplotsset{compat=1.7}
\usepackage[normalem]{ulem}
\usepackage{tikz}
\usetikzlibrary{arrows,positioning,quantikz2}
\usepackage{graphicx}
\usepackage{wrapfig}
\usepackage{csquotes}
\usepackage{etoolbox}
\usepackage{mathpartir}     % for inference rules
\usepackage{listings}
\usepackage[usenames,dvipsnames]{xcolor}
\usepackage{makecell}
\usepackage{amsmath}
\usepackage{enumitem}
\setlist{nosep,leftmargin=\parindent,}
\usepackage{mathrsfs}
\usepackage{float}
\usepackage{hyperref}

\usepackage{upgreek}
\usepackage{adjustbox}
\usepackage{multirow}
\usepackage{stmaryrd}
\usepackage{xspace}
\usepackage[framemethod=default]{mdframed}
\mdfsetup{skipabove=.67\topskip,skipbelow=.67\topskip}
\usepackage{xcolor}
\definecolor{CYAN}{rgb}{0,1,1}

\newtoggle{TR}
\toggletrue{TR}

\usepackage{cleveref}
\crefname{section}{\S}{\S\S}
\Crefname{section}{\S}{\S\S}
\crefformat{section}{\S#2#1#3}
\iftoggle{TR}{}{%
  \crefname{appendix}{Appendix}{Appendices}
  \Crefname{appendix}{Appendix}{Appendices}
  \crefname{subappendix}{Appendix}{Appendices}
  \Crefname{subappendix}{Appendix}{Appendices}
}

\usepackage{mathpartir}

\usepackage{amssymb}
\newtheorem{thm}{Theorem}[section]

\newtheorem{lem}{Lemma}[section]
\newtheorem{defn}{Definition}[section]

\newcommand {\qbar} {{\overline{q}}}
\newcommand {\tr} {{\mathrm{Tr}}}
\newcommand{\SWAP}{\textit{SWAP}}

\DeclareMathOperator{\supp}{supp}

\newcommand{\sem}[1]{\llbracket #1 \rrbracket}

\newcommand{\op}[2]{|#1\rangle \langle #2|}
\newcommand{\ip}[2]{\langle #1|#2\rangle }

\newcommand{\CNOT}{\textit{CNOT}}
\newcommand{\SKIP}{\ensuremath{\mathbf{skip}}}

\newcommand{\logic}{\textsc{SAQR-QC}\xspace}

\newcommand{\naturals}{\mathbb{N}}

\newcommand{\ketbra}[2]{\left| #1 \right\rangle \!\left\langle #2 \right|}

\newcommand{\twr}[1]{{\color{magenta}{T: #1}}}
\newcommand{\twrchanged}[1]{{\color{cyan}{#1}}}

\newcommand{\ny}[1]{{\color{ForestGreen}{#1}}}

\newcommand{\revision}[1]{{\color{ForestGreen}{#1}}}

\iftoggle{TR}{}{%
  \AtBeginDocument{%
    \makeatletter
    \newlabel{sec:ProofsOfTwoLemmas}{{A}{26}{Lemmas and Proofs of Lemmas}{appendix.A}{}}
    \newlabel{sec:ReductionPrincipleProof}{{B}{26}{}{appendix.B}{}}
    \newlabel{sec:SoundnessProof}{{C}{27}{}{appendix.C}{}}
    \newlabel{sec:Eq20ImpliesEq15}{{D}{27}{}{appendix.D}{}}
    \newlabel{sec:AlternativeGHZ}{{E}{28}{}{appendix.E}{}}
    \newlabel{sec:QFTviaQAI}{{F}{29}{}{appendix.F}{}}
    \newlabel{sec:ProofOfTwoEquations}{{G}{30}{}{appendix.G}{}}
    \newlabel{sec:ProofOfC1}{{G.1}{30}{}{subsection.G.1}{}}
    \newlabel{sec:ProofOfC3}{{G.2}{31}{}{subsection.G.2}{}}
    \newlabel{sec:ProofsOfTwoLemmas@cref}{{[appendix][1][]A}{[1][26][]26}{}{}{}}
    \newlabel{sec:ReductionPrincipleProof@cref}{{[appendix][2][]B}{[1][26][]26}{}{}{}}
    \newlabel{sec:SoundnessProof@cref}{{[appendix][3][]C}{[1][27][]27}{}{}{}}
    \newlabel{sec:Eq20ImpliesEq15@cref}{{[appendix][4][]D}{[1][27][]27}{}{}{}}
    \newlabel{sec:AlternativeGHZ@cref}{{[appendix][5][]E}{[1][28][]28}{}{}{}}
    \newlabel{sec:QFTviaQAI@cref}{{[appendix][6][]F}{[1][29][]29}{}{}{}}
    \newlabel{sec:ProofOfTwoEquations@cref}{{[appendix][7][]G}{[1][30][]30}{}{}{}}
    \newlabel{sec:ProofOfC1@cref}{{[subappendix][1][7]G.1}{[1][30][]30}{}{}{}}
    \newlabel{sec:ProofOfC3@cref}{{[subappendix][2][7]G.2}{[1][31][]31}{}{}{}}
    \makeatother
  }
}

\begin{document}
\title[\logic: A Logic for Scalable but Approximate Quantitative Reasoning about Quantum Circuits]{\logic: A Logic for Scalable but Approximate \\ Quantitative Reasoning about Quantum Circuits}

\author{Nengkun Yu}
\orcid{0000-0003-1188-3032}
\affiliation{%
  \institution{Stony Brook University}
  \city{Stony Brook}
  \country{USA}
}
\email{nengkun.yu@cs.stonybrook.edu}

\author{Jens Palsberg}
\orcid{0000-0003-4747-365X}
\affiliation{%
  \institution{University of California at Los Angeles}
  \city{Los Angeles}
  \country{USA}
}
\email{palsberg@ucla.edu}

\author{Thomas Reps}
\orcid{0000-0002-5676-9949}
\affiliation{%
  \institution{University of Wisconsin-Madison}
  \city{Madison}
  \country{USA}
}
\email{reps@cs.wisc.edu}

%
% \titlerunning{\logic: A Logic for Quantitative Reasoning about Quantum Circuits}

%
\input{abstract}

\begin{CCSXML}
<ccs2012>
   <concept>
       <concept_id>10003752.10003753.10003758</concept_id>
       <concept_desc>Theory of computation~Quantum computation theory</concept_desc>
       <concept_significance>500</concept_significance>
       </concept>
   <concept>
       <concept_id>10003752.10003790.10002990</concept_id>
       <concept_desc>Theory of computation~Logic and verification</concept_desc>
       <concept_significance>500</concept_significance>
       </concept>
 </ccs2012>
\end{CCSXML}

\ccsdesc[500]{Theory of computation~Quantum computation theory}
\ccsdesc[500]{Theory of computation~Logic and verification}

\keywords{Verification, quantitative reasoning for quantum circuits}

\maketitle              % typeset the header of the contribution

\input{intro}

\input{preliminary}

\input{QAI}
\input{Predicates}

\input{GHZ}
\input{QPE}
\input{related}

\input{conclusion}

\begin{acks}
Supported, in part, by NSF grant OMA-2016245 and DARPA grant HR00112520301.
\end{acks}

\newpage
\bibliographystyle{ACM-Reference-Format}
\bibliography{main}

\iftoggle{TR}{%
  \newpage
  \appendix
  \input{appendix}

}{%
}

\end{document}

%% file: abstract.tex
\begin{abstract}

Reasoning about quantum programs remains a fundamental challenge, regardless of the programming model or computational paradigm.
Existing verification techniques are insufficient---even for quantum circuits, a deliberately restricted model that lacks classical control, but still underpins many current quantum algorithms. Many existing formal methods require exponential time and space to represent and manipulate (representations of) assertions and judgments, making them impractical for quantum circuits with many qubits. This paper presents \logic, a logic for \textbf{S}calable but \textbf{A}pproximate \textbf{Q}uantitative \textbf{R}easoning about \textbf{Q}uantum \textbf{C}ircuits. \logic has three characteristics:
(i) some deliberate loss of precision is built into it;
(ii) it has a mechanism to help the accumulated loss of precision during a sequence of reasoning steps remain small; and
(iii)
every reasoning step is local---involving just a small number of qubits---making reasoning scalable.
We demonstrate the effectiveness of \logic via two case studies:
the verification of GHZ circuits involving non-Clifford gates, and
the analysis of quantum phase estimation---a core subroutine in Shor’s factoring algorithm.
% As a key application, we develop scalable methods for the static analysis of quantum programs---a domain where quantitative reasoning is essential due to the probabilistic nature of quantum algorithms.
% Our approach combines Ying’s Quantum Hoare Logic, using quantum observables as predicates, with a quantum abstract interpretation framework based on local projections as abstract states. 

% I think we should say that both the deliberate-loss-of-precision aspect and the use of two interacting formalisms are inspired by abstract interpretation.

% \keywords{Quantitative Reasoning \and Scalable Verification \and Quantum Programming Languages}
\end{abstract}

%% file: intro.tex
\section{Introduction}
\label{sec:intro}

Quantum computing leverages superposition and interference to achieve computational advantages over classical methods for specific tasks. For instance, Shor’s algorithm~\cite{Shor94} efficiently factors large integers using the Quantum Fourier Transform (QFT) and Quantum Phase Estimation (QPE) \cite{Kitaev1995}.
QPE is also a foundational component in quantum-simulation algorithms~\cite{Aspuru-Guzik2005} and in quantum algorithms for solving linear-algebra problems~\cite{HHL08}.
% Importantly, the quantum advantage in these algorithms often arises from components that do not rely on classical control-flow constructs, such as loops or conditional statements, highlighting the potential for local, scalable reasoning about their behavior.

Guaranteeing that quantum programs behave as intended is a fundamental challenge.
This challenge has spurred significant research into verifying quantum programs using classical computers \cite{BHY19,ZYY19,zhou2021quantum,10.1145/3453483.3454029,Unruh,YP21,10.1145/3704873,10.1145/3519939.3523431,CACM:CCLLTY25}. A landmark achievement is Ying's Quantum Hoare Logic (QHL)~\cite{Ying11}, which extends classical Hoare logic to reason about the correctness of quantum programs. QHL uses quantum predicates---semidefinite positive Hermitian operators---as pre- and post-conditions, with 
\revision{
a
}
quantum Hoare triple \cite{DP06} defined as:
\begin{align}\label{eq:qhl}
  \{A\}\ \mathbf{C}\ \{B\} \quad \mathrm{iff} \quad \textrm{for all input state}~\rho, \tr(A\rho) \leq \tr(B \sem{\mathbf{C}}(\rho)).
\end{align}
Here, $\operatorname{Tr}(A \rho)$ represents the probability (or expected value) that the input state $\rho$ satisfies the predicate $A$, while
$\operatorname{Tr}\!\bigl(B\, \llbracket \mathbf{C} \rrbracket(\rho)\bigr)$ represents the probability that the output state satisfies $B$ after executing $C$.
When $\{A\}\ \mathbf{C}\ \{B\}$ holds, the program $C$ ensures that the probability of satisfying the postcondition $B$ is at least as great as the probability of satisfying the precondition $A$.

Despite the development of QHL, scalable verification remains elusive, even for quantum circuits---a restricted model without classical control that encapsulates many current algorithms. Existing methods struggle with exponential complexity; for instance, while QHL is complete for unitary programs, its application requires computing \( U_\mathbf{C}^\dagger B U_\mathbf{C} \) for a general \( n \)-qubit unitary \( U_\mathbf{C} \), a task requiring time and space exponential in \( n \). This intractability motivates our central question:
\noindent
\begin{mdframed}[innertopmargin=3pt, innerbottommargin=3pt]
  \textit{Can we establish a theoretical foundation for \textbf{scalable quantitative reasoning} about quantum programs that use many qubits?}
\end{mdframed}
\noindent
Here, ``scalable'' means that the size of a proof---the matrices
involved and the logical derivation
itself---grows polynomially in the number of qubits.

This paper answers this question affirmatively by introducing \logic\ (\textbf{S}calable but \textbf{A}pproximate \textbf{Q}uantitative \textbf{R}easoning for \textbf{Q}uantum \textbf{C}ircuits), a logic for reasoning about programs expressed as quantum circuits. \logic is designed with three principles:
(i) it embraces a deliberate, controlled loss of precision to achieve scalability;
(ii) it incorporates mechanisms to keep the accumulated imprecision small across reasoning steps;
(iii) all reasoning is \emph{local}, meaning each step involves only a constant number of qubits, independent of the total system size.

\logic is inspired by two prior approaches that sit at opposite ends of a spectrum.
At one end, QHL---see \Cref{eq:qhl}---offers fully precise quantitative reasoning, but is computationally intractable for large systems.
At the other end, Quantum Abstract Interpretation (QAI) \cite{YP21} provides a scalable, qualitative framework. A QAI judgment,
\begin{align*} % \label{eq:qai}
  \models^{\rm QAI}\{\mathcal{P}\}\ \mathbf{C}\ \{\mathcal{Q}\} \quad \mathrm{iff} \quad 
  \mathrm{for all}\,\rho,\, \rho\vDash \mathcal{P}~\text{implies}~ \sem{\mathbf{C}}(\rho)\vDash \mathcal{Q},
\end{align*}
uses tuples of local projections $\mathcal{P}$ and $\mathcal{Q}$ as abstract states. QAI enables efficient, local reasoning but lacks the ability to reason about quantitative properties like success probabilities.

\logic integrates ideas from both to achieve scalable quantitative reasoning.
We wish to stress that \logic is a \emph{logic for manual proof construction}, analogous to how Hoare Logic, Separation Logic, and Linear Logic were first developed as conceptual frameworks to facilitate human-driven reasoning.
Automation, while a vital future direction, is not the aim of this paper.

Each assertion that appears in a \logic judgment has a two-part structure, e.g., $\{\mathscr{A} \mid \mathcal{P}\}$, inspired by the spatial and pure assertions of separation logic \cite{DBLP:conf/fmco/BerdineCO05,DBLP:conf/tacas/DistefanoOY06}.
$\mathscr{A}$ is a tuple of \emph{local observables}, each acting non-trivially on only a constant number of qubits (i.e., the number does not grow as the circuit size increases).
These observables track quantitative information, such as success probabilities---similar to QHL---but restricted to local components.
$\mathcal{P}$ is a QAI-style tuple of local projections, capturing qualitative spatial constraints on the state.
In essence, $\{\mathscr{A} \mid \mathcal{P}\}$ gives us a way of specifying a collection of quantum states
% (or, equivalently, a collection of reduced density matrices)
via a tuple of local observables and a tuple of local projections.
\logic uses judgments of the form
\begin{align}\label{eq:saqr-judgment}
  \{\mathscr{A} \mid \mathcal{P}\}\ \mathbf{C}\ \{\mathscr{B} \mid \mathcal{Q}\},
\end{align}
% 
% \revision{
% Each assertion that appears in a \logic judgment
% }
% has a two-part structure,
% \revision{
% e.g., $\{\mathscr{A} \mid \mathcal{P}\}$,
% }
% inspired by the spatial and pure assertions of separation logic \cite{DBLP:conf/fmco/BerdineCO05,DBLP:conf/tacas/DistefanoOY06}.
% \revision{
% In a \logic judgement
% }
% \begin{align}\label{eq:saqr-judgment}
%   \{\mathscr{A} \mid \mathcal{P}\} \mathbf{C} \{\mathscr{B} \mid \mathcal{Q}\},
% \end{align}
% \( \mathcal{P} \) and \( \mathcal{Q} \) are QAI-style tuples of local projections, capturing qualitative spatial constraints on the state. \( \mathscr{A} \) and \( \mathscr{B} \) are tuples of \emph{local observables}, each acting non-trivially on only a constant number of qubits.
% \revision{
% (By ``constant number,'' we mean that the number does not grow as the circuit size increases.)
% }
% These observables track quantitative information, such as success probabilities, in a manner reminiscent of QHL, but restricted to local components.
% 
which asserts that for any input state \( \rho \) satisfying \( \mathcal{P} \), 
(i) the output state \( \sem{\mathbf{C}}(\rho) \) satisfies \( \mathcal{Q} \), and (ii) a certain quantitative relationship involving \( \mathscr{A} \) and \( \mathscr{B} \) holds.
More precisely, judgment \labelcref{eq:saqr-judgment} means
\begin{align}
  \models^{\rm \logic} \{\mathscr{A} \mid \mathcal{P}\}\ \mathbf{C}\ \{\mathscr{B} \mid \mathcal{Q}\}
  \hspace{1.5ex} \mathrm{iff} \hspace{1.5ex}
  \textrm{for all} ~\rho,\, \rho\vDash \mathcal{P}~\textrm{implies}~~
     \begin{aligned}[t]
       & \textrm{(i)}~\sem{\mathbf{C}}(\rho)\vDash \mathcal{Q},~~\mathrm{and} \\
       & \textrm{(ii)}~\tr(M_\mathscr{A}\rho)\leq \tr(M_\mathscr{B}\sem{\mathbf{C}}(\rho)),
     \end{aligned}
\end{align}
where $M_\mathscr{A}$ and $M_\mathscr{B}$ are derived from $\mathscr{A}$ and $\mathscr{B}$, with the precise definition provided in \Cref{Se:CorrectnessFormulasAndTheLogicalSystem}.
Here, $\tr(M_{\mathscr{A}}\rho)$ denotes the total weight (or expectation value) with which the input state $\rho$ satisfies the local observables specified by $\mathscr{A}$, while $\tr(M_{\mathscr{B}}\sem{\mathbf{C}}(\rho))$ denotes the corresponding weight for the output state $\sem{\mathbf{C}}(\rho)$ with respect to $\mathscr{B}$.
When $\{\mathscr{A} \mid \mathcal{P}\}\ \mathbf{C}\ \{\mathscr{B} \mid \mathcal{Q}\}$ holds, we restrict our attention to input states that satisfy $\mathcal{P}$---i.e., $\rho \vDash \mathcal{P}$---and on those states, the program $\mathbf{C}$ ensures that the output state $\sem{\mathbf{C}}(\rho)$ satisfies $\mathcal{Q}$ and that the degree of satisfaction of $\mathscr{B}$ is at least that of $\mathscr{A}$.

This two-part structure is not merely stylistic; it is essential under the locality constraint. Without locality, the two components could be combined into a single QHL triple \cite[Theorems 3.2 and 3.3]{ZYY19}. However, with the restriction to local projections and observables, this combination is generally impossible, justifying the distinct roles of the QAI-like and QHL-like components. Their interaction---inspired by the reduced product in abstract interpretation \cite[\S10.1]{DBLP:conf/popl/CousotC79}, as well as the use of abstract interpretation in logic-based tools~\cite{Crab}---helps to
control the accumulation of imprecision.

\textbf{\textit{Overview.}}
The form of \Cref{eq:saqr-judgment} is motivated by the need for scalable reasoning.
A logic for quantum computing must support sequences of reasoning steps.
For instance, for sequential composition---i.e.,
for input state $\rho_0$,
$\sem{\mathbf{C}_1;\mathbf{C}_2}(\rho_0)
=
\sem{\mathbf{C}_2} (\sem{\mathbf{C}_1}(\rho_0))$---we aim to derive a judgment $\{\mathscr{A}\}\,\mathbf{C}_1;
\mathbf{C}_2 \,\{\mathscr{C}\}$ from (derivations of) judgments of the form $\{\mathscr{A}\}\,\mathbf{C}_1\,\{\mathscr{B}\}$ and $\{\mathscr{B}\}\,\mathbf{C}_2 \,\{\mathscr{C}\}$.
To ensure that we can bound the sizes of $\mathscr{A}$, $\mathscr{B}$, and $\mathscr{C}$ (where $\mathscr{B}$ describes the set of intermediate states $\{ \sem{\mathbf{C}_1}(\rho_0) \mid \rho_0 \models \mathscr{A} \}$, and
$\mathscr{C}$ describes $\{ \sem{\mathbf{C}_2}(\rho) \mid \rho \models \mathscr{B} \}$), we must use a less-general language than that permitted in QHL:
we require the logic to use constraints that have two properties: (i) they are efficiently representable, and (ii) satisfaction ($\rho \models \varphi$) is efficiently computable.
% One existing approach that satisfies these requirements is Quantum Abstract Interpretation (QAI)~\cite{YP21}.

Let $\mathscr{A} := (A_{s_1}, \cdots, A_{s_m})$ and $\mathscr{B} := (B_{s_1}, \cdots, B_{s_m})$ 
be tuples of local observables, where each $A_{s_i}$ and $B_{s_i}$ acts only on a subset of qubits 
$s_i \subseteq \{1,2,\cdots,n\}$.
The tuple $(s_1, \cdots, s_m)$ is fixed in advance, and typically each $s_i$ is chosen independently of the total number of qubits $n$.
This choice defines a class of observables efficiently represented by structures, such as vectors and matrices, whose size grows at most polynomially in the number of qubits:
reasoning about a sub-circuit that acts only on qubits 1 and 2 need only concern those qubits, not the full state space.
% The motivation for restricting to such observables can be illustrated by a simple example: 
% consider circuits that consist only of unitary operations acting on qubits $1$ and $2$. 
% In this case, it suffices to reason about the behavior of the first two qubits, and thus we only need to consider the subset $\{1,2\}$.

In \Cref{Se:CorrectnessFormulasAndTheLogicalSystem}, \Cref{simpleexample} shows that merely restricting QHL to use local observables yields overly imprecise reasoning, which motivates combining local observables $\mathscr{A}$ with the local projections $\mathcal{P}$ from QAI---giving rise to predicates of the form $\{ \mathscr{A} \mid \mathcal{P} \}$ and to triples as in \Cref{eq:saqr-judgment}.

A primary use case of \logic{} is forward reasoning: given pre-state assertion \(\{\mathscr{A}\mid \mathcal{P}\}\) and circuit $\mathbf{C}$, derive a post-state assertion \(\{\mathscr{B}\mid \mathcal{Q}\}\).
Exploiting the compositional nature of quantum circuits, and noting that any circuit can be decomposed into two-qubit unitary operations, it suffices for \logic{} to handle the case where \(\mathbf{C}\) is a two-qubit unitary \(U\). 
In this setting, the derivation of \(\mathscr{B}\) from a given \(\mathscr{A}\) reduces to updating only those local observables in \(\mathscr{A}\) that are relevant to the two qubits on which \(U\) acts.
$\mathcal{Q}$ is derived from $\mathcal{P}$ by the local-reasoning rules of QAI.
In principle, $\mathscr{B}$ can also be influenced by $\mathcal{P}$;
such ``cross-talk'' can be seen in our example \Cref{simpleexample}.
The resulting post-condition may not be the strongest possible; however, it is designed to be sufficiently informative to extract useful quantitative properties, such as the success probability of a computation in the BQP model.\footnote{
  BQP is the class of decision problems solvable by a quantum computer in polynomial time with bounded error, analogous to the class P for classical deterministic computation.
}

\textbf{\textit{Limitations.}}
A fundamental limitation of \logic concerns expressivity.
Its assertion language is based on tuples of local observables and tuples of local projections, making it well suited to reasoning about properties that can be decomposed into local components. 
Consequently, global properties that inherently depend on non-local correlations may not be directly expressible.

Compared with QHL, which allows arbitrary Hermitian operators as predicates, \logic is strictly less expressive: QHL can specify global properties of quantum states, including those involving entanglement across an unbounded number of qubits, whereas \logic is restricted to properties that can be captured via fixed collections of local observables and local projections.
This restriction is essential for scalability, but comes at the cost of expressivity.

On the other hand, \logic is more expressive than existing instances of QAI, which rely on coarse abstract domains. 
In particular, \logic can capture quantitative relationships among multiple observables and track their evolution through circuit transformations, enabling reasoning about properties such as success probabilities in BQP computations. 
Such properties are typically beyond the reach of standard QAI domains, which focus on qualitative
properties.

More broadly, the expressivity of \logic is closely related to the question of how global quantum properties can be reconstructed from local information. 
This connection suggests that techniques from abstract interpretation—particularly results on the inherent limits of abstract domains and their reduced products—may offer useful perspectives. 
At the same time, the quantum setting involves fundamentally different mathematical structures, and it is not yet clear to what extent these classical insights carry over. 
We view this as an important issue that warrants further investigation in the development of \logic.
We discuss how to address “failures” (i.e., from overly loose approximations) in \Cref{PrecisionandMitigation}.

As with any manual logic, constructing a \logic proof requires insight, particularly in choosing the appropriate local projections and observables. The logic provides the framework, but the user must supply the clever insights---the ``eureka'' steps---to define the predicates that make the verification possible. The case studies in \Cref{Se:GHZ} and \Cref{Se:qpe} exemplify this  principle, demonstrating how specific, carefully chosen predicates yield sharp results for non-trivial quantum circuits.

\textbf{\textit{Contributions.}}
This paper lays a foundation for scalable reasoning about quantum computations.
% by defining a logic that is both quantitative and tractable, and by demonstrating its power on representative quantum algorithms:
\begin{itemize}
  \item We present \logic, which integrates quantitative, QHL-like reasoning with the scalable, local reasoning of QAI.
  The design of \logic ensures that all reasoning steps are local, leading to proofs whose size is polynomial in the number of qubits and gates.
  
  \item We provide a formal foundation for \logic, defining its assertions and proof rules. 
  \item We demonstrate \logic's utility and expressiveness through two detailed case studies:
    \begin{itemize}
      \item \emph{Verification of GHZ circuits involving non-Clifford gates}: We derive a judgment that characterizes the output state precisely in every aspect except a relative phase factor.
      \item \emph{Analysis of QPE}:
      We derive a judgment showing that for any constant \( k \), the QPE algorithm provides the best estimate of the last \( k \) bits of the phase, with probability at least \( 4/\pi^2 \). A key step is a novel, lossless local-reasoning method for the QFT using QAI. To the best of our knowledge, no prior approach has achieved such a result with a proof that scales polynomially with the system size.
    \end{itemize}
  \item We clarify the relationship between \logic and prior work, particularly Zhou et al.\ \cite{ZYY19}, and discuss the principles for constructing effective predicates for use in our framework.
\end{itemize}

\textbf{\textit{Organization.}}
\Cref{sec:quantumbackground} presents background material about quantum computing, as well as the basics of QAI.
\Cref{Se:CorrectnessFormulasAndTheLogicalSystem} defines local observables as predicates, and formalizes the judgments of \logic, which integrate local observables and QAI.
We employ \logic to reason about the general GHZ circuit (\Cref{Se:GHZ}) and quantum phase estimation (\Cref{Se:qpe}).
\Cref{Se:RelatedWorks} discusses related work.
\Cref{sec:conclusion} concludes the paper and outlines future work.
\iftoggle{TR}{
  Some proofs and derivations are available in 
  Appendices \ref{sec:ProofsOfTwoLemmas}-\ref{sec:ProofOfTwoEquations}.
}{
  References to Appendices \ref{sec:ProofsOfTwoLemmas}-\ref{sec:ProofOfTwoEquations} refer to the appendices of \cite{yu2025logic}.
}

%% file: preliminary.tex
%!TEX root = main.tex
%!TEX spellcheck = en_US

\section{Background \& Notation}
\label{sec:quantumbackground}

%To make the paper self-contained, this section provides background material, and discusses notation for---and properties of---quantum computing. The material in this section is similar to what can be found in published books and papers, e.g., \cite{NI11,YP21,Ying11}. Readers already familiar with quantum computing may wish to proceed directly to \Cref{quali}.

%To make the paper self-contained, this section reviews basic concepts, notation, and properties of quantum computing, following standard references such as \cite{NI11,YP21,Ying11}, as well as the basics of QAI. Readers familiar with these topics may wish to skip to \Cref{Se:CorrectnessFormulasAndTheLogicalSystem}.

To make the paper self-contained, this section briefly reviews basic quantum computing concepts and QAI (following \cite{NI11,YP21,Ying11});
some readers may wish to skip directly to \Cref{Se:CorrectnessFormulasAndTheLogicalSystem}.

\textbf{\textit{Preliminaries.}}
% \label{sec:preliminaries}
%We use the notation $[n]$ to denote the set $\{1, \ldots, n\}$, \( \setminus \) to denote set difference, and \( |s| \) to denote the cardinality of a set \( s \). We assume familiarity with Dirac notation and standard linear algebra concepts, including Hilbert space,  tensor products, orthonormal bases, and inner and outer products.
We write $[n] = \{1, \dots, n\}$, $\setminus$ for set difference, and $|s|$ for the cardinality of a set $s$. We assume familiarity with Dirac notation, $\ket{\cdot}$, and standard linear-algebra concepts, including Hilbert spaces, tensor products, orthonormal bases, and inner/outer products.

Linear \emph{operators} on $d$-dimensional complex vector spaces are represented by $d \times d$ matrices $\mathbb{C}^{d \times d}$. The identity is $I$, and the conjugate transpose of $A$ is $A^\dag = (A^T)^*$.
An operator is \emph{Hermitian} if $A=A^\dag$ and \emph{positive semi-definite} if all eigenvalues are nonnegative; its trace is $\tr(A)=\sum_i A_{ii}$.
%An operator is \emph{Hermitian} if $A = A^\dag$ and \emph{positive semi-definite} if all eigenvalues are non-negative. The trace, $\tr(A)$, is the sum of the diagonal entries: $\tr(A) = \sum_i A_{ii}$.

The \emph{L\"{o}wner order} on Hermitian matrices, $A \le B$ if $B-A$ is positive semidefinite, is fundamental for comparing states and operators in quantum mechanics.

%Linear \emph{operators} on \( d \)-dimensional complex vector spaces are represented as \( d \times d \) matrices over \( \mathbb{C} \), denoted \( \mathbb{C}^{d \times d} \). The identity operator is denoted by \( I \). For an operator \( A \), the conjugate transpose is defined as \( A^\dag = (A^T)^* \), where \( A^T \) is the transpose and \( (\cdot)^* \) denotes complex conjugation. An operator \( A \) is \emph{Hermitian} if \( A = A^\dag \), and \emph{positive semi-definite} if all its eigenvalues are non-negative. The trace of a matrix \( A \), denoted \( \tr(A) \), is the sum of its diagonal entries: \( \tr(A) = \sum_i A_{ii} \).

%The \emph{L\"{o}wner order}, defined on Hermitian matrices in quantum mechanics and convex analysis, establishes a relationship \( A \leq B \) indicating that \( B - A \) is positive semidefinite. The Löwner order is vital for comparing states and operators.

\textbf{\textit{Quantum States.}}
%\label{sec:quantum-states}
A \emph{quantum state} describes the state of a quantum system. A single qubit \emph{pure} state $\ket{\psi}$ lies in a two-dimensional Hilbert space as a superposition of $\ket{0}$ and $\ket{1}$. An $n$-qubit system resides in a $2^n$-dimensional space, allowing complex superpositions and entanglement. \emph{Mixed states}, represented by a density matrix $\rho$, generalize pure states to probabilistic mixtures. %Quantum operations include \emph{unitaries}, which preserve probability, and \emph{measurements}, which probabilistically collapse the system to classical outcomes.

%A \emph{quantum state} describes the state of a quantum system. For a single qubit, the state \( \ket{\psi} \) belongs to a two-dimensional Hilbert space and can exist as a superposition of the basis states \( \ket{0} \) and \( \ket{1} \). For an \( n \)-qubit system, the state resides in a \( 2^n \)-dimensional Hilbert space and can exhibit both complex superpositions and entanglement among qubits.

%Quantum systems may also be in \emph{mixed states}, represented by a \emph{density matrix} \( \rho \), which generalizes pure states to account for probabilistic mixtures---capturing both classical and quantum uncertainty. Quantum operations, whether on pure or mixed states, include \emph{unitary transformations} (which preserve total probability) and \emph{measurements}, which probabilistically collapse the system to a classical outcome based on the state's amplitudes.

\textbf{\textit{Reduced Density Matrices.}}
%\label{sec:reduced-density-matrices}
Reduced density matrices are central to analyzing multipartite quantum systems, as many properties depend solely on subsystem reductions. In quantum computation, for example, the success probability of algorithms like HHL~\cite{HHL08} depends only on the reduced density matrix of the ``signal'' qubit. Measuring this qubit, success (outcome $\ket{1}$) depends only on its reduced state, independent of global entanglement or structure.
%Measuring this qubit, the computation succeeds if the outcome is $\ket{1}$, and this probability is fully determined by its reduced state, independent of the global state's entanglement or structure.

Let $\mathbb{C}^{d_1}$ and $\mathbb{C}^{d_2}$ be the Hilbert spaces of two quantum systems. The composite system lives in $\mathbb{C}^{d_1} \otimes \mathbb{C}^{d_2}$, and analyzing subsystems uses the \emph{partial trace}. The partial trace over $\mathbb{C}^{d_1}$, $\tr_1(\cdot)$, maps operators on $\mathbb{C}^{d_1} \otimes \mathbb{C}^{d_2}$ to $\mathbb{C}^{d_2}$:
$\tr_1\big( \ket{\varphi_1}\bra{\psi_1} \otimes \ket{\varphi_2}\bra{\psi_2} \big) = \braket{\psi_1 | \varphi_1} \, \ket{\varphi_2}\bra{\psi_2}$, for all $\ket{\varphi_1}, \ket{\psi_1} \in \mathbb{C}^{d_1}$ and $\ket{\varphi_2}, \ket{\psi_2} \in \mathbb{C}^{d_2}$, extended linearly.
Similarly, $\tr_2(\cdot)$ traces out $\mathbb{C}^{d_2}$.
%For a composite density matrix $\rho$, the reduced states are $\tr_2(\rho)$ and $\tr_1(\rho)$.  

For an $n$-qubit system and $s \subseteq [n]$, the reduced density matrix is \(\rho_s = \tr_{[n] \setminus s}(\rho),
\)
with $\tr_{[n] \setminus s}$ tracing out all qubits not in $s$. The partial trace preserves positive semi-definiteness~\cite{NI11}.

\textbf{\textit{Unitary Operations.}}
%\label{sec:QTL}
%Unitary operations are fundamental transformations in quantum mechanics; they preserve the norm of a quantum state and are represented by a unitary matrix. A unitary matrix \( U \) satisfies \( U^\dag U = I \), where \( U^\dag \) denotes the conjugate transpose of \( U \) and \( I \) is the identity matrix. These operations are crucial for manipulating quantum states and implementing quantum algorithms. For a pure state \( \ket{\psi} \), applying a unitary operator \( U \) transforms \( \ket{\psi} \) to \( U\ket{\psi} \). For a density operator \( \rho \), the transformation is \( \rho \mapsto U\rho U^\dag \).
%
Unitary operations, represented by matrices $U$ with $U^\dag U = I$, preserve the norm of quantum states and are fundamental for manipulating them and implementing algorithms. They act on pure states as $\ket{\psi} \mapsto U\ket{\psi}$ and on density operators as $\rho \mapsto U \rho U^\dag$. Commonly used single-qubit operators include
the Pauli gates \( I \), \( X \), \( Y \), and \( Z \);
the Hadamard gate \( H \); the \( T \) gate;
the family of gates $\set{R_m \mid m \in \naturals}$.
Commonly used two-qubit gates include the SWAP operation \( \text{SWAP} \) and the controlled-NOT operation \( \text{CNOT} \).
\begin{small}
\begin{gather*}
I = \begin{pmatrix} 1 & 0 \\0 & 1 \end{pmatrix} \qquad
X = \begin{pmatrix} 0 & 1 \\ 1 & 0 \end{pmatrix} \qquad
Y = \begin{pmatrix} 0 & -i \\ i & 0 \end{pmatrix} \qquad
Z = \begin{pmatrix} 1 & 0 \\ 0 & -1 \end{pmatrix} \qquad
H = \frac{1}{\sqrt{2}}\begin{pmatrix} 1 & 1\\ 1 & -1 \end{pmatrix}
\\
T = \begin{pmatrix} 1 & 0 \\ 0 & e^{i\pi/4} \end{pmatrix} \ \qquad R_{m}={\begin{pmatrix}1&0\\0&e^{2\pi i/2^{m}}\end{pmatrix}} \qquad\ \ \ 
\SWAP = \begin{pmatrix}
1 & 0 & 0 & 0 \\
0 & 0 & 1 & 0 \\
0 & 1 & 0 & 0 \\
0 & 0 & 0 & 1
\end{pmatrix}
 \ \qquad
\CNOT = \begin{pmatrix}
1 & 0 & 0 & 0 \\
0 & 1 & 0 & 0 \\
0 & 0 & 0 & 1 \\
0 & 0 & 1 & 0
\end{pmatrix}
\end{gather*}
\end{small}

\textbf{\textit{Observables.}}
%\twr{Don't we need a subsection that reviews observables?}
%Quantum observables are physical quantities in a quantum system that can be measured, such as position, momentum, energy, and spin. Mathematically, observables are represented by \emph{Hermitian} (self-adjoint) operators on a Hilbert space, satisfying \( O^\dagger = O \), which ensures that measurement outcomes are real. In certain contexts---such as \emph{Quantum Hoare Logic} (QHL)---it is common to further restrict observables to the range \( 0 \leq O \leq I \), where the inequality is understood in the \emph{Löwner partial order} (\Cref{sec:preliminaries}), meaning that both \( O \) and \( I - O \) are positive semidefinite. These bounded observables often serve as predicates or quantum effects, capturing partial truth values within the quantum program-verification framework.
%
Quantum observables represent measurable quantities, such as position, energy, or spin, and are modeled by \emph{Hermitian} operators $O$ that satisfy $O^\dag = O$, which ensures real-valued measurement outcomes. In contexts like QHL, observables are often restricted to $0 \le O \le I$ in the \emph{L\"{o}wner order}, meaning both $O$ and $I-O$ are positive semidefinite.
These observables serve as predicates or quantum effects, capturing partial truth values in program verification.

\textbf{\textit{Quantum Circuits and Semantics.}}
We consider quantum programs, represented as circuits on $n$ qubits, composed of $p$ unitary gates $U_{f_1}, \dots, U_{f_{p}}$, where each $U_{f_\ell}$ acts on a subset of qubits $f_\ell \subseteq [n]$. The program starts in $\ket{0}^{\otimes n} = \ket{0^n}$, and its semantics is given by $U_{f_{p}} \cdots U_{f_1} \ket{0^n}$.

Each gate is lifted to an $n$-qubit unitary by tensoring with identities. For a single-qubit gate $U$ acting on qubit $i$,
$U \otimes I_{[n]\setminus\{i\}} := (\otimes_{k>i} I) \otimes U \otimes (\otimes_{0 < j < i} I)$,
and similarly for two-qubit gates $U_{f_\ell}$ on $f_\ell=\{i,j\}$, $U \otimes I_{[n]\setminus\{i,j\}}$,
with placement determined by the qubit indices.  
For clarity, we describe a program's semantics assuming 2-qubit gates, although all results extend to gates acting on up to $m$ qubits for any constant $m$.

%In this work, we focus on quantum programs---represented as quantum circuits---operating on a fixed number of qubits, says $n$.
%A quantum program is composed of a sequence of unitary instructions \( U_{F_1}, \dots, U_{F_{|p|}} \). Each gate \( U_{F_\ell} \) operates on a subset of qubits \( F_\ell \subseteq [n] \).
%The initial state is taken to be $\ket{0}^{\otimes n} = \ket{0^n}$, and the meaning (or semantics) of the program is given by the matrix product \[ U_{F_{|p|}} \cdots U_{F_1} \ket{0^n}. \]
%To interpret \( U_{F_\ell} \) as an \( n \)-qubit unitary, we embed it into the full register by tensoring with identity operators. If \( F_\ell = \{i\} \), then the corresponding lifted unitary is
%\[
%U \otimes I_{[n] \setminus \{i\}} := (\otimes_{k>i} I) \otimes U \otimes (\otimes_{0 \leq j < i} I),
%\]
%where \( U \) acts nontrivially on qubit \( i \), and \( I \) is the identity matrix on one qubit. We will explicitly subscript identity matrices to indicate which qubits they act upon.

%For a two-qubit unitary \( U_{F_\ell} \) acting on qubits \( F_\ell = \{i,j\} \), we similarly interpret it as
%\[
%U \otimes I_{[n] \setminus \{i,j\}},
%\]
%with the appropriate placement determined by the positions of \( i \) and \( j \) in the register.

%To simplify the presentation and reduce notational overhead, we formulate the semantics of quantum circuits assuming 2-qubit gates. Nevertheless, all of our methods naturally extend to circuits composed of gates acting on up to \( m \) qubits, for any constant \( m \).

\begin{definition}[Syntax] \label{def:syntax}
The syntax of quantum programs is given by
\[
\mathbf{C} ::= \SKIP \mid \bar{q} := U[\bar{q}] \mid \mathbf{C}_1; \mathbf{C}_2
\]
\end{definition}

We write \( \sem{\mathbf{C}} \) to denote the semantics of a quantum program \( \mathbf{C} \). If \( \mathbf{C} \) represents a unitary transformation \( U_C \), then for any input density matrix \( \rho \), its semantics is given by
$\sem{\mathbf{C}}(\rho) := U_C \rho U_C^\dag$.
We also define the dual action on observable $A$ as:
$\sem{\mathbf{C}}^*(A) := U_C^\dag A U_C$.
We define
\( \mathcal{U} := \lambda x.U x U^\dagger \) to denote the quantum operation on density matrices that corresponds to the unitary matrix \( U \) (which operates on quantum states). That is, \( \mathcal{U} \rho=\lambda x.U x U^\dagger \rho= U \rho U^\dag.\)

\textbf{\textit{Projections.}}
%\label{sec:background-linear-algebra}
%An orthogonal projection matrix \( P \) satisfies \( P = P^\dag = P^2 \), which a stricter condition than the classical \( P = P^2 \).For short, we refer to such matrices as ``projections.'' For example, \( \op{00}{00} + \op{11}{11} \) is a rank-2 projection that projects any 4-dimensional vector onto a 2-dimensional subspace.
%
%Each projection \( P \) corresponds to a unique subspace \( S_P = \{ v \mid Pv = v \} \), and we use the terms ``projections'' and ``subspaces'' interchangeably. This correspondence establishes a partial order: for projections \( P \) and \( Q \), we have \( P \subseteq Q \) iff \( S_P \subseteq S_Q \).
%
%Projections are positive semi-definite. The support of a positive semi-definite matrix \( A \), \( \supp(A) \), is the subspace spanned by eigenvectors with nonzero eigenvalues. According to Birkhoff and von Neumann \cite{BvN36}, a density matrix \( \rho \) satisfies a projection \( P \), denoted by \( \rho \vDash P \), if \( \supp(\rho) \subseteq P \). This property is equivalent to \( P\rho = \rho \).
%
An orthogonal projection satisfies $P = P^\dag = P^2$, a stronger condition than the classical $P = P^2$. We simply call such matrices \emph{projections}. Each projection $P$ corresponds to a subspace $S_P = \{v \mid Pv = v\}$, and we use ``projection'' and ``subspace'' interchangeably.
Viewing projections as subspaces induces a partial order: $P \subseteq Q$ iff $S_P \subseteq S_Q$. For example, $\op{00}{00} + \op{11}{11}$ is a rank-2 projection onto a 2D subspace of $\mathbb{C}^4$.

Projections are positive semidefinite. The support of a positive semidefinite matrix $A$, $\supp(A)$, is the span of eigenvectors with nonzero eigenvalues. A density matrix $\rho$ \emph{satisfies} $P$, written $\rho \vDash P$, if $\supp(\rho) \subseteq P$, equivalently $P\rho = \rho$~\cite{BvN36}.

\textbf{\textit{Lemmas.}}
Our development of \logic relies on three fundamental operations on operators:
\begin{itemize}
  \item \textbf{Löwner order} of operators (denoted \( A \leq B \)),
  \item \textbf{Partial trace and trace operators} (denoted \( \tr_s \) and \( \tr \), where \( \tr_s \) traces out subsystem \( s \)),
  \item \textbf{Expansion} of an operator via tensor product (denoted \( A_s \otimes I_{[n] \setminus s} \), where the operator \( A_s \) acts on subsystem \( s \) and is expanded to the full system).
\end{itemize}

%The following lemmas capture algebraic relationships among these operations. These lemmas are instrumental in establishing the correctness of our \logic framework in a purely algebraic style. In what follows, we assume \( s\subseteq [n] \); that \( P \) is a projection operator on \( n \)-qubit systems; that \( A,B \) is a positive semidefinite matrix on an \( n \)-qubit space, $E$ be a matrix; and that \( \rho \) is a quantum state.

The following lemmas summarize key algebraic relationships among these operations, forming the basis for the correctness of our \logic framework. Unless stated otherwise, let $s \subseteq [n]$; $P$ be a projection on an $n$-qubit system; $A,B$ be positive semidefinite matrices; $E$ a matrix; and $\rho$ a quantum state. The proof of \Cref{rdm} is given in \Cref{sec:ProofsOfTwoLemmas}.
The other lemmas can be proven using the definition of trace, partial trace, and support.

The relevance of these lemmas to \logic is that
(i) a positive semidefinite matrix $A$ serves as a predicate on a density matrix $\rho$ via the expectation value $\tr(A\rho)$; and
(ii) a subclass of such predicates with good locality properties are ones defined by a positive semidefinite matrix $A_s$ that only acts on a qubit set $s \subseteq [n]$, producing the expectation value $\tr\left( (A_s \otimes I_{[n] \setminus s}) \rho \right) = \tr\left( A_s \rho_s \right)$.

\iffalse
\begin{lem} \label{lem:sum-of-states}
$\supp(A_1+A_2)=\supp(A_1)\vee\supp(A_2):=\mathrm{span}\{=\supp(A_1),\supp(A_2)\}$. In particular, $\supp(A_1)\subseteq P$ and $\supp(A_2)\subseteq P$ if and only if
$\supp(A_1+A_2)\subseteq P$.
\end{lem}

\begin{lem}[\cite{YP21}]
\label{lem:partial-trace-support}
$\supp(\tr_{s}(\supp(A)))=\supp(\tr_{s}A)$.
\end{lem}

\begin{lem}[\cite{YP21}]
\label{lem:partial-trace-comuting}
If $s_i\cap s_j=\emptyset$, then $\tr_{s_j}(\tr_{s_i} A)=\tr_{s_i}(\tr_{s_j} A)=\tr_{s_i\cup s_j} A$.
\end{lem}

\begin{lem}
\label{lem:extension-with-identity-matrix}
If $\supp(A)\subseteq P$, then $\supp(A\otimes Q)\subseteq P\otimes Q$ and
$\supp(\tr_{s} A) \subseteq \supp(\tr_s P)$.
\end{lem}

\begin{lem}[\cite{YP21}]
\label{lem:partial-trace-support}
$\supp(\tr_{s}(\supp(A)))=\supp(\tr_{s}A)$.
\end{lem}

\begin{lem}[The second part is from \cite{YP21}]
\label{lem:trace-preservation-under-U-and-U-dag}
Let \( \mathbf{C} \) be a quantum circuit that acts only on qubits in the subset \( s \subseteq [n] \). Then for any observable \( A \) on the full system:
\begin{align*}
\tr_s\big((\sem{\mathbf{C}} \otimes \mathcal{I})(A)\big) &= \tr_s A, \\
\tr_{[n] \setminus s}\big((\sem{\mathbf{C}} \otimes \mathcal{I})(A)\big) &= \sem{\mathbf{C}} \big( \tr_{[n] \setminus s} A \big).
\end{align*}
\end{lem}
\fi
\begin{lem}\label{rdm}
Let \( \rho \) be the density matrix of an \( n \)-qubit system, and let \( s \subseteq [n] \). Then for any observable \( A_s \) acting on subsystem \( s \), $\tr\left( (A_s \otimes I_{[n] \setminus s}) \rho \right) = \tr\left( A_s \rho_s \right)$.
\end{lem}

\begin{lem} \label{lem:interproduct}
For two square matrices $B$ and $E$ of the same size,
$\tr(BE)=\tr(EB)$.
\end{lem}

\begin{lem} \label{lem:tr_ge_zero}
For $A,B\geq 0$, \(\tr(AB)\geq 0\).
\end{lem}

\begin{lem} \label{lem:tr_monotonic_in_Loewner_order}
If \( A \le B \), then (i) \( \tr(A\rho) \le \tr(B\rho) \) for any density operator \( \rho \), and in particular \( \tr(A) \le \tr(B) \); (ii) \( P A P^{\dagger} \le P B P^{\dagger} \) for any operator \( P \).
\end{lem}

\begin{lem} \label{lem:satisfyprojection}
For a density matrix \( \rho \) and a projection \( P \), we have
\[
\tr(P \rho) = 1 \quad \Leftrightarrow \quad \supp(\rho) \subseteq P \quad \Leftrightarrow \quad \rho \vDash P.
\]
\end{lem}

%% file: QAI.tex
\textbf{\textit{Qualitative Predicates for Local Reasoning: Quantum Abstract Interpretation.}}
% \label{quali}
%
%Qualitative predicates serve as logical assertions about quantum states, such as whether a subsystem lies in a given subspace. When expressed as tuples of local projectors, they enable scalable \emph{local reasoning}, inferring global behavior from partial system views. Quantum Abstract Interpretation (QAI)~\cite{YP21} systematically propagates such predicates through quantum circuits using projectors on small subsystems, avoiding the exponential cost of full-state analysis.
%
%This section reviews local projective predicates, formalizes their semantics, and describes their transformation under unitaries. Soundness is established via support-based semantics and partial trace, forming a foundation for abstract reasoning in quantum programs.
%
The qualitative predicates used in QAI express logical properties of quantum states, e.g., whether a subsystem lies in a subspace. Represented as tuples of local projectors, they support scalable local reasoning, inferring global behavior from partial views.
QAI~\cite{YP21} propagates these predicates through circuits using small-subsystem projectors, avoiding exponential full-state analysis.

This section reviews local projective predicates, formalizes their semantics, and describes their transformation under unitaries. Soundness is established via support-based semantics and partial trace, forming the basis for abstract reasoning
about quantum programs.
``Locality'' means that we work with a tuple of sets $(s_1, \dots, s_m)$, where each $s_i \subseteq [n]$ is a small subset of bounded size.

\begin{defn}[\cite{YP21}]\label{projectivepredicates}
A tuple \( (P_{s_1}, \cdots, P_{s_m}) \) is called a \emph{projective predicate} if each \( P_{s_i} \) is a projection, i.e., \( P_{s_i}^2 = P_{s_i} \). We use \( \mathcal{P} \) (or \( \mathcal{Q}, \mathcal{R} \)) to denote projective predicates. In particular, we write \( \mathcal{I} := (I_{s_1}, \cdots, I_{s_m}) \) to represent the identity predicate.
\end{defn}
\textbf{Remark.} Projective predicates are also referred to as \emph{abstract states}.

\begin{defn}[\cite{YP21}]\label{projectivepredicates-concrete}
A state $\rho$ satisfies a projective predicate $\mathcal{P} = (P_{s_1}, \cdots, P_{s_m})$, denoted by $\rho \vDash^{QAI} \mathcal{P}$, if for all $1 \le i \le m$, $P_{s_i}\rho_{s_i} = \rho_{s_i}$, i.e., $\rho_{s_i} \vDash P_{s_i}$. Equivalently, $\rho \vDash \gamma(\mathcal{P})$, where
\[
  \gamma(\mathcal{P}) := \bigcap_i P_{s_i} \otimes I_{[n] \setminus s_i}.
\]
\end{defn}

Given quantum circuit \( \mathbf{C} \) and state \( \rho \vDash \mathcal{P} \), QAI \cite{YP21} constructs a predicate \( \mathcal{Q} \) such that the post-state \( \sem{\mathbf{C}}(\rho) \) satisfies \( \mathcal{Q} \), denoted by $\vDash^{QAI} \{\mathcal{P}\}{ \mathbf{C} }\{\mathcal{Q}\}$.
The idea---encapsulated in the following theorem---is to perform partial concretization rather than complete concretization.

\begin{thm}[\cite{YP21}]\label{QAI}
Let \( U_F \) be a unitary gate applied to the qubit set \( s(F) \), and let \( \mathcal{P} = (P_{s_1}, \cdots, P_{s_m}) \) be a projective predicate. For each \( s_i \), define
\[
\begin{array}{@{\hspace{0ex}}l@{\hspace{2.75ex}}c@{\hspace{2.75ex}}r@{\hspace{0ex}}}
  R_i = \bigcap_{\substack{s_j \subseteq s_i \cup s(F)}} P_{s_j} \otimes I_{s_i \cup s(F) \setminus s_j},
  &
  Q_{s_i} = \supp\left( \tr_{s_i \cup s(F) \setminus s_i} \left( U_F R_i U_F^\dag \right) \right),
  &
   U^\sharp(\mathcal{P}) = (Q_{s_1}, \cdots, Q_{s_m}). 
\end{array}
\]
Then
$\qquad
\rho \vDash^{QAI} \mathcal{P} \quad \Rightarrow \quad
U_F \rho U_F^\dag \vDash^{QAI} U^\sharp(\mathcal{P}).
$
\end{thm}

\begin{thm}[\cite{YP21}]\label{assertion}
Let \( P = \mathrm{span}\{ \ket{a_1 a_2 \cdots a_{n}}, \ket{b_1 b_2 \cdots b_{n}} \} \),  
where the product states \( \ket{a_i} \) and \( \ket{b_i} \) are not proportional for every \( i \in [n] \).  
Then $P = \gamma(\mathcal{P})$,
where \( \mathcal{P} = (P_{1,2}, \ldots, P_{n-1,n}) \) and each  
\(P_{i,i+1} = \mathrm{span}\{ \ket{a_i a_{i+1}}, \ket{b_i b_{i+1}} \}.
\)
\end{thm}

%% file: Predicates.tex
\section{Correctness Formulas and the Logical System \logic}
\label{Se:CorrectnessFormulasAndTheLogicalSystem}

In this section, we introduce a class of predicates that can be used for quantitative local reasoning (\Cref{Se:predicates}).
We then sketch a strawman approach that incorporates these predicates in a logic similar to QHL (\Cref{sec:FirstAttemptAtDefiningJudgments}).
Via a simple example, we demonstrate that this strawman approach can lead to a significant loss of precision in reasoning.
To address this limitation, we integrate the QAI technique into the framework, leading to the formal definition of the judgments used in \logic (\Cref{sec:QuantitativeJudgmentsAndValidity}),
a theorem about how \logic judgments relate to judgments in QAI (\Cref{Se:Reduction}),
a systematic presentation of the inference rules of \logic (\Cref{Se:LogicalSystemWithSoundness}), and an approximation strategy that addresses how to work with the one inference rule of \logic that poses a challenge to scalability (\Cref{methods}).

\subsection{Quantitative Local Reasoning via Generalized Predicates}\label{Se:predicates}\label{Se:correctness}

Reasoning about full quantum states quickly becomes infeasible due to the exponential growth of the state space with the number of qubits.
To address this issue, we want an approach akin to QAI~\cite{YP21}, which uses tuples of projections as predicates that capture properties of small subsystems.
However, QAI itself is unsatisfactory because QAI can only perform \emph{qualitative} local reasoning.

To enable \emph{quantitative} local reasoning, we focus on \emph{reduced density matrices}, which capture the behavior of local subsystems.
Inspired by Quantum Hoare Logic (QHL)—where a global positive semidefinite matrix $A$ serves as a \emph{predicate} describing properties of a state $\rho$ via the expectation value $\tr(A\rho)$—we introduce the notion of \emph{local observables} to reason about subsystem states $\rho_s$.

\begin{defn}\label{De:LocalObservable}
A \emph{local observable} over a qubit set $s\subseteq [n]$ is a positive semidefinite operator $A_s$ satisfying $0 \leq A_s \leq I_s$, acting nontrivially only on the qubits in $s$. 
\end{defn}

According to Lemma~\ref{rdm}, $\tr(A_s \rho_s) = \tr(A \rho)$ with $A := A_s \otimes I_{[n]\setminus s}$. 
Local observables thus form a natural subclass of predicates in QHL, and we use them to monitor reduced density matrices via the inner product $\tr(\cdot)$, in direct analogy with QHL;
here, ``monitoring $\rho_s$'' means computing the expectation value of local observable $A_s$ with respect to $\rho_s$---i.e., computing $\tr(A_s \rho_s)$.

To capture richer information, we track not a single $\rho_s$ but a tuple of reduced density matrices $(\rho_{s_1}, \dots, \rho_{s_m})$, where each $s_i \subseteq [n]$ is a small subset of bounded size
(i.e., in the same spirit as QAI).
% We also refer to
% \twrchanged{
% the $\rho_{s_i}$ as \emph{low-dimensional marginals}. 
% }
% \twr{We can remove the previous sentence: we never use the term ``low-dimensional marginals'' anywhere else.}
This representation (i) provides significantly deeper insight into the system’s local structure, and (ii) remains tractable, because it requires only \emph{linear} resources, i.e., scales linearly with $m$.

The idea is to define a quantum predicate for each \( s_i \), represented by an observable \( A_{s_i} \), focusing solely on the state of the quantum registers in \( s_i \).

\begin{defn}\label{predicates}
A \emph{predicate} over an $n$-qubit space and $S=(s_1,\dots,s_m)$ is a tuple of local observables $(A_{s_1},\dots,A_{s_m})$. with $0 \le A_{s_i} \le I_{s_i}$, each acting non-trivially only on $s_i$. Predicates are denoted by $\mathscr{A}$ (or $\mathscr{B},\mathscr{D}$);
$m$ is the predicate's \emph{size}.
The domain of $(A_{s_1}, \dots, A_{s_m})$, denoted by $\mathrm{dom}(A_{s_1}, \dots, A_{s_m})$, is $(s_1, \dots, s_m)$.
\end{defn}

Each local observable acts nontrivally on just a few qubits, and captures the expectation value of a measurable quantity, such as the success probability, fidelity, or entanglement,
% Local observables act nontrivially on few qubits and capture expectation values corresponding to measurable quantities, such as success probabilities, fidelities, and entanglement,
providing a direct link between semantics and operational outcomes. By tracing their evolution under unitaries and partial traces, we obtain a \emph{scalable, compositional} framework for quantum program verification, supporting precise analysis without reconstructing the full state.

These generalized predicates track quantitative aspects of quantum programs---such as probabilities and expectation values---by monitoring reduced density matrices over small, constant-size subsystems. This abstraction underpins scalable, compositional reasoning about quantum computations via low-dimensional summaries.

\begin{defn}\label{matrix}
For a predicate $\mathscr{A} = (A_{s_1}, \cdots, A_{s_m})$, its matrix representation $M_{\mathscr{A}}$ is defined as:
\begin{small}
\begin{equation}\label{eq:MatrixRepresentation}
M_{\mathscr{A}} = \sum_{i=1}^m A_{s_i} \otimes I_{[n] \setminus s_i}.
\end{equation}
\end{small}
\end{defn}
This definition induces a function $\tr(M_{\mathscr{A}}\rho)$ on a state $\rho$. By Lemma~\ref{rdm}, the value of $\tr(M_{\mathscr{A}}\rho)$ depends only on the reduced states $\rho_{s_i}$ and is equal to the sum of the expectations obtained by measuring $A_{s_i}$ on each $\rho_{s_i}$:
\begin{align}\label{keyobservation}
\tr(M_{\mathscr{A}}\rho) 
= \tr \Big( \sum_{i=1}^m (A_{s_i}\otimes I_{[n]\setminus s_i})\rho \Big) 
= \sum_{i=1}^m \tr[(A_{s_i}\otimes I_{[n]\setminus s_i})\rho] 
= \sum_{i=1}^m \tr (A_{s_i}\rho_{s_i}).
\end{align}
% \begin{defn}
% The domain of a predicate $(A_{s_1}, \dots, A_{s_m})$ is $\mathrm{dom}(A_{s_1}, \dots, A_{s_m}) = (s_1, \dots, s_m)$.
% \end{defn}
%In particular, the matrix in \Cref{matrix} specifies how a (static) set of quantum states is characterized in terms of expectations of local observables: $\{ \rho \mid \tr(M_{\mathscr{A}}\rho)\geq r \}$ for any given $r\geq 0$.}
%\twr{Is the following a clearer restatement what we intend to say in the previous sentence?
%}
Moreover, matrix  $M_{\mathscr{A}}$ in \Cref{eq:MatrixRepresentation} implicitly specifies a set of quantum states in terms of how their expectations of local observables relate to a given threshold $r$:
$\{ \rho \mid \tr(M_{\mathscr{A}}\rho) \geq r \}$, for $r\geq 0$.

To reason about program behavior, we additionally require a notion of predicate transformation that captures how such predicates evolve under quantum operations. 
In \logic{}, this notion is realized via linear predicate transformers, which map pre-state predicates to post-state predicates. 
Operationally, these transformers describe how the local observables—and hence the corresponding linear functions—are updated by a given circuit element.
%\twr{Why ``linear \emph{functionals}'' and not ``linear \emph{functions}''?Can we explain the connection between linear functionals and expectations of local observables?
%}

This separation between representation (predicates) and transformation (predicate transformers) enables scalable, compositional quantitative reasoning.
% \twr{The phrase ``this definition tracks changes'' in the previous sentence suggests that something in \Cref{matrix} talks about tracking changes.
% But a predicate characterizes a (static) set of states, not \emph{changes} to states, so the matrix defined in \Cref{matrix} is just saying how a set of states is represented.
% The issue of ``tracking changes''---i.e., tracking the \emph{dynamic} effect of an operation on a set of states---is outside the realm of this subsection.
% Unfortunately, we forgot to explain this whole concept: we need to say what predicate transformers/linear functionals are, and to answer the questions ``how are predicates transformed?'' and `` how do we specify/represent a predicate transformer?''
% }

We are also interested in the order relation between predicates.

\begin{defn}\label{PredicateOrder}
Given predicates $(A_{s_1}, \cdots, A_{s_m})$ and $(A_{s_1}', \cdots, A_{s_m}')$ over the same domains, define
\[
(A_{s_1}, \cdots, A_{s_m}) \sqsubseteq (A_{s_1}', \cdots, A_{s_m}') \quad \text{if and only if} \quad A_{s_i} \leq A_{s_i}' \text{ for all } i.
\]
\end{defn}

\begin{lem}\label{mono}
The matrix representation of predicates is monotonic with respect to this ordering---i.e.,
\[
\mathscr{A} \sqsubseteq \mathscr{B} \quad \Rightarrow \quad M_{\mathscr{A}} \leq M_{\mathscr{B}}.
\]
\end{lem}

%Of course, the projective predicates in \Cref{quali} can also be regarded as predicates, where under the latter scenario, the interpretation will be for the linear product, but not whether \twr{they?} live in this subspace.
%\twr{The last sentence needs to be clarified.}

\subsection{A First Attempt at Defining Judgments}
\label{sec:FirstAttemptAtDefiningJudgments}

This section presents an initial attempt for \emph{quantitative local reasoning} in quantum programs using tuples of local observables.

In QHL~\cite{Ying11}, the unitary rule gives the \emph{weakest precondition} for a postcondition:
$
  \{ U_\mathbf{C} ^\dag B U_\mathbf{C}  \}\  \mathbf{C} \ \{ B \}.
$
To ensure that postcondition $B$ is efficiently representable and physically relevant, we restrict $B$ to the kind of predicate introduced in \Cref{matrix}:
we consider a tuple of local observables, where locality indicates that each observable acts nontrivially on only a constant number of qubits.
Let $S = (s_1, \dots, s_m)$, with $s_i \subseteq [n]$, denote the tuples of local qubits that are of interest, and let $\mathscr{B} = (B_{s_1}, \dots, B_{s_m})$ with \( 0 \leq B_{s_i} \leq I_{s_i} \); then
$
M_{\mathscr{B}} = \sum_i B_{s_i} \otimes I_{[n] \setminus s_i}.
$
The question now is ``What is the analogue of $U_\mathbf{C} ^\dag B U_\mathbf{C}$?''
\begin{defn}\label{LHP}
Given a fixed domain $S = (s_1, \dots, s_m)$ and a quantum program 
$\mathbf{C} = \lambda \rho.\, U \rho U^\dag$ on a density matrix, 
we say that $\mathscr{A} = (A_{s_1}, \dots, A_{s_m})$ is a \emph{local precondition} 
of $\mathscr{B} = (B_{s_1}, \dots, B_{s_m})$ if
\[
  \models \{M_{\mathscr{A}}\}\, \mathbf{C} \, \{M_{\mathscr{B}}\},
\]
where $M_{\mathscr{A}} = \sum_i A_{s_i} \otimes I_{[n]\setminus s_i}$ and 
$M_{\mathscr{B}} = \sum_i B_{s_i} \otimes I_{[n]\setminus s_i}$, 
and the judgment means
\[
  \forall \rho, \quad \tr(M_{\mathscr{A}} \rho) \le \tr(M_{\mathscr{B}} \sem{\mathbf{C}}(\rho)).
\]
\end{defn}

We could use a QHL-like proof rule for a unitary, namely 
\begin{equation}
  \inferrule
    { M_{\mathscr{A}} \leq U^\dag M_{\mathscr{B}} U  }
    { \{ M_{\mathscr{A}} \}\ \mathbf{C}\ \{ M_{\mathscr{B}} \} }
  \label{Eq:LocalBackwardReasoning}
\end{equation}
To see that $M_{\mathscr{A}} \leq U^\dag M_{\mathscr{B}} U$ leads to a valid QHL judgement---without requiring that the observables are bounded above by the identity---we observe that
\begin{equation}
  \label{Eq:DeducedTraceInequality}
  \tr\left(M_{\mathscr{A}} \rho\right) 
  \leq \tr\left(U^\dag M_{\mathscr{B}} U \rho\right) 
  = \tr\left(M_{\mathscr{B}} U \rho U^\dag\right) 
  = \tr\left(M_{\mathscr{B}}\, \sem{\mathbf{C}}(\rho)\right),
\end{equation}
where the inequality follows from \Cref{lem:tr_monotonic_in_Loewner_order} applied to the premise of Rule~\labelcref{Eq:LocalBackwardReasoning}, and the first equality uses \Cref{lem:interproduct}.
\Cref{LHP} tells us that the conclusion of Rule~\labelcref{Eq:LocalBackwardReasoning} holds.
 
This definition naturally gives rise to a correctness judgment of the form \( \{ \mathscr{A} \}\ \mathbf{C}\ \{ \mathscr{B} \} \) for quantum circuits.
In addition to enabling \emph{backward reasoning}, Rule~\labelcref{Eq:LocalBackwardReasoning} with \Cref{LHP} also supports \emph{forward reasoning}
by changing the premise of Rule~\labelcref{Eq:LocalBackwardReasoning} to
``$U M_{\mathscr{A}} U^\dag \leq M_{\mathscr{B}}$.''

Unfortunately, it is not possible, in general, to use 
either version of Rule~\labelcref{Eq:LocalBackwardReasoning}
\emph{algorithmically} to compute the weakest precondition and the strongest postcondition, respectively.
The issue is one of \emph{expressibility}:
can the answer be decomposed according to the chosen scheme $S = (s_1, \ldots, s_m)$?
In general, this \emph{structural constraint} of our candidate logic presents an obstacle.
For instance,
% \emph{Local reasoning imposes structural constraints that prevent directly saturating the inequality to compute the} \emph{strongest postcondition} or \emph{weakest precondition}:
even when \( U_F \) is a two-qubit unitary, the transformed observable
$
  U_F^\dag \left( \sum_i B_{s_i} \otimes I_{[n] \setminus s_i} \right) U_F
$
may not admit a decomposition of the form \( \sum_i A_{s_i} \otimes I_{[n] \setminus s_i} \) for any choice of local observables \( A_{s_i} \). That is, the following equality may have \emph{no solution} for local predicates \( \mathscr{A} \), for some given \( \mathscr{B} \).
\begin{align} 
\label{equality}
M_{\mathscr{A}} = U_F^\dag M_{\mathscr{B}} U_F
\end{align}

\begin{example}\label{simpleexample}
To understand this limitation concretely, consider a circuit in which a CNOT gate is applied to the input state $\ket{00}$, as shown in Figure~\ref{fig:CNOT}.
The output state is $\ket{00}$.
Let us reason about the circuit \(\mathbf{C}\) using \Cref{equality}, where we choose \((s_1, s_2) = (\{q_1\}, \{q_2\})\).
Let the postcondition be $\mathscr{B} = (B_{s_1}, B_{s_2}) = (\op{0}{0},\op{0}{0})$, which has the matrix representation \(M_{\mathscr{B}} = \op{0}{0} \otimes I + I \otimes \op{0}{0}\).
\begin{mdframed}[innertopmargin=3pt, innerbottommargin=3pt]
  \textit{\textbf{Observation}: The following calculation will demonstrate that, in local reasoning, it is not always possible to express the weakest precondition.}
\end{mdframed}

\noindent
We will show that there is no \( A_1, A_2 \) such that
\vspace{-2pt}
  \begin{align*}
    U_{\mathbf{C}}^\dag \left( \sum_i B_{s_i} \otimes I_{[n] \setminus s_i} \right) U_{\mathbf{C}} = \op{0}{0}\otimes I + \op{0}{0}\otimes \op{0}{0} + \op{1}{1}\otimes \op{1}{1} = A_1 \otimes I + I \otimes A_2
  \end{align*}

We see that the left-hand side is orthogonal to $\op{1}{1}\otimes \op{0}{0}$. Therefore, we know that

\noindent
\noindent
\begin{minipage}[b]{0.68\textwidth}
\vspace{0pt}
\[
\begin{aligned}
& \tr[(\op{1}{1}\otimes \op{0}{0})(A_1\otimes I+I\otimes A_2)] = 0 \\
\Rightarrow\;& \tr[(\op{1}{1}A_1)\otimes \op{0}{0}] + \tr[\op{1}{1}\otimes(\op{0}{0}A_2)] = 0 \\
\Rightarrow\;& \tr(\op{1}{1}A_1)=0 \quad \text{and} \quad \tr(\op{0}{0}A_2)=0 \ \ \ \qquad (\textrm{Lemma} \ \ref{lem:tr_ge_zero})\\
\Rightarrow\;& A_1=\lambda_1 \op{0}{0}, \quad A_2=\lambda_2 \op{1}{1} \\
\Rightarrow\;& A_1 \otimes I + I \otimes A_2
\end{aligned}
\]
\end{minipage}\hfill
\begin{minipage}[t]{0.28\textwidth}
\vspace{-2.5em}  % <-- adjust this (2em, 3em, 4em)
\centering
\includegraphics[width=\linewidth]{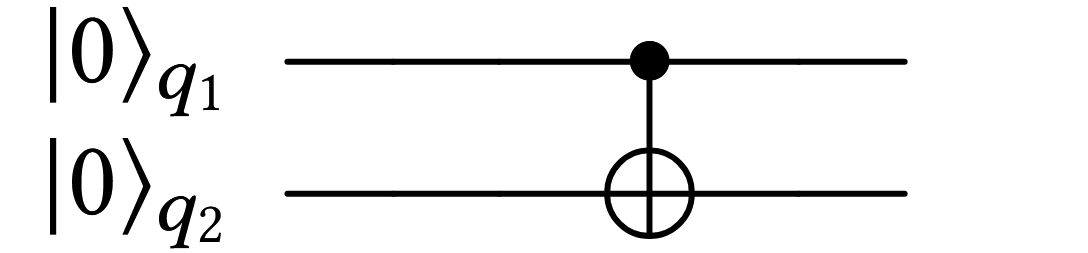}
\captionof{figure}{CNOT circuit \(\mathbf{C}\).}
\label{fig:CNOT}
\end{minipage}
Moreover, we can show that the QHL-like proof rule \labelcref{Eq:LocalBackwardReasoning} is sometimes unable to prove a
simple assertion.
For instance, for the input $\ket{00}$, the $\textit{CNOT}$ gate has no effect;
therefore, the output state is also $\ket{00}$.
Consequently,
the following assertion holds:
\begin{align}
  \label{Eq:SimpleAssertion}
  \{(\op{0}{0},\op{0}{0})\}\ \mathbf{C}\ \{(\op{0}{0},\op{0}{0})\}.
\end{align}
To establish that Rule~\labelcref{Eq:LocalBackwardReasoning} cannot prove that Assertion~\labelcref{Eq:SimpleAssertion} holds,
we compute
\begin{align*}
                      &\{(\op{0}{0},\op{0}{0})\}\ \mathbf{C}\ \{(\op{0}{0},\op{0}{0})\}\\
  \Longleftrightarrow & \mathit{CNOT}(\op{0}{0}\otimes I+I\otimes\op{0}{0})\mathit{CNOT}^{\dag}\leq \op{0}{0}\otimes I+I\otimes\op{0}{0}\\
  \Longleftrightarrow &\op{0}{0}\otimes I+\op{00}{00}+\op{11}{11}\leq \op{0}{0}\otimes I+I\otimes\op{0}{0}\\
  \Longleftrightarrow &\op{00}{00}+\op{11}{11}\leq I\otimes\op{0}{0}
\end{align*}
Because the final inequality fails, in our strawman logic based on tuples of local observables and Rule~\labelcref{Eq:LocalBackwardReasoning}, Assertion~\labelcref{Eq:SimpleAssertion} is unprovable;
i.e.,
\begin{align*}
  \not\vdash\{(\op{0}{0},\op{0}{0})\}\ \mathbf{C}\ \{(\op{0}{0},\op{0}{0})\}.
\end{align*}

What this example shows is that without additional constraints that restrict the input, the candidate logic lacks precision.
As we show in \Cref{simpleexample_redux} (see \Cref{sec:QuantitativeJudgmentsAndValidity}), by imposing the constraint $\mathcal{P} := (\op{0}{0}, \op{0}{0})$ on the input, \logic can establish the desired property that input state $\ket{00}$ yields output state $\ket{00}$.
%\twr{Now add some $\mathcal{P}$ to restrict the initial set of states, and show that we get a much better answer.}

% We can further characterize the local precondition of $(\op{0}{0},\op{0}{0})$ as follows,
%  \begin{align*}
% &\{(A_1,A_2)\}\mathbf{C}\{(\op{0}{0},\op{0}{0})\}\\
% \Longleftrightarrow &A_1\otimes I+I\otimes A_2\leq \mathit{CNOT}^{\dag}(\op{0}{0}\otimes I+I\otimes\op{0}{0})\mathit{CNOT}\\
% \Longleftrightarrow & A_1\otimes I+I\otimes A_2\leq \op{0}{0}\otimes I+\op{00}{00}+\op{11}{11}\\
% \Longrightarrow & \tr_1[(A_1\otimes I+I\otimes A_2)(\op{1}{1}\otimes I)]\leq \tr_1[(\op{0}{0}\otimes I+\op{00}{00}+\op{11}{11})(\op{1}{1}\otimes I)]\\
% \Longrightarrow & \tr(A_1 \op{0}{0})I+A_2\leq \op{1}{1}\\
% \Longrightarrow & A_2\leq \op{1}{1}
% \end{align*}
% Let $A_2=\lambda \op{1}{1}$ with $0\leq \lambda\leq 1$, we have
%  \begin{align*}
% &A_1\otimes I+I\otimes A_2\leq \op{0}{0}\otimes I+\op{00}{00}+\op{11}{11}\\
% \Longleftrightarrow & A_1\otimes I\leq 2\op{00}{00}+ (1-\lambda)I\otimes\op{1}{1}\\
% \Longleftrightarrow & A_1\leq 2\op{0}{0} {\quad\textrm{and}\quad A_1 \leq (1-\lambda)I} \\
% \Longleftrightarrow & A_1\leq (1-\lambda)\op{0}{0}.
% \end{align*}
% In summary, we can only prove
% \begin{align*}
% \vdash\{((1-\lambda)\op{0}{0},\lambda\op{1}{1})\}\mathbf{C}\{(\op{0}{0},\op{0}{0})\} \ \forall 0\leq \lambda\leq 1
% \end{align*}
% by using \Cref{LHP} alone.

% This simple example demonstrates that directly using Rule~\labelcref{Eq:LocalBackwardReasoning} to analyze quantum programs may result in a significant loss of information.
\end{example}

\subsection{Quantitative Judgments and Validity}
\label{sec:QuantitativeJudgmentsAndValidity}

Beyond the expressivity limitations noted in \Cref{sec:FirstAttemptAtDefiningJudgments}, the predicates used in \Cref{sec:FirstAttemptAtDefiningJudgments} cannot even capture the singleton state $\ket{0}^{\otimes n}$, revealing a mismatch with our goals: 
\Cref{LHP} demands validity for all states satisfying the precondition, 
whereas quantum programs typically concern the specific input $\ket{0}^{\otimes n}$. 
To address both issues, we adopt a more expressive logic that still supports quantitative local reasoning, 
leveraging quantum abstract interpretation~\cite{YP21} as a foundation.

\begin{defn}\label{mdef}
A judgment is a triple of the form $\{\mathscr{A} \mid \mathcal{P}\}\ \mathbf{C}\ \{\mathscr{B} \mid \mathcal{Q}\}$ for program $\mathbf{C}$, general predicates $\mathscr{A}$ and $\mathscr{B}$, and projective predicates $\mathcal{P}$ and $\mathcal{Q}$.
\end{defn}

For simplicity, we assume that $\mathscr{A}$, $\mathscr{B}$, $\mathcal{P}$, and $\mathcal{Q}$ share the same domain.

\begin{defn}[Validity] \label{quantitative-validity-total}
The judgment $\{\mathscr{A} \mid \mathcal{P}\}\, \mathbf{C}\,\{\mathscr{B} \mid \mathcal{Q}\}$ in \Cref{mdef} holds if \begin{equation}
\label{eq:judge}
\forall \rho \vDash^{QAI} \mathcal{P},\;
\sem{\mathbf{C}}(\rho)\vDash^{QAI}\mathcal{Q},\;
\text{and}\;
\tr(M_{\mathscr{A}}\rho)\le
\tr(M_{\mathscr{B}}\sem{\mathbf{C}}(\rho)).
\end{equation}
where $M_{\mathscr{A}}$ and $M_{\mathscr{B}}$ are the matrix representations of $\mathscr{A}$ and $\mathscr{B}$, as
defined in
\Cref{matrix}.

Consider the
special case $\mathcal{P}=\mathcal{Q}=\mathcal{I}$ with $\mathcal{I}:=(I_{s_1},\cdots,I_{s_m})$, we will write 
\[
\{\mathscr{A} \mid \mathcal{I}\}\ \mathbf{C}\ \{\mathscr{B} \mid \mathcal{I}\}=\{\mathscr{A}\}\ \mathbf{C}\ \{\mathscr{B}\}
\]
in which case this definition simplifies to \Cref{LHP}.
\end{defn}

% We have that ${\tr}(M_{\mathscr{A}}\rho)$ only depends on the reduced density matrices of $\rho$
% with respect to the systems $\set{s_i}$,
As we know from \Cref{keyobservation}, $\tr(M_{\mathscr{A}} \rho) = \sum_{i=1}^m \tr\left( A_{s_i} \rho_{s_i} \right)$.
In other words, ${\tr}(M_{\mathscr{A}}\rho)$ only depends on the reduced density matrices of $\rho$ with respect to  the systems $\set{s_i}$.
Intuitively, this judgment tracks a linear function of the tuple of reduced density matrices to enable quantitative reasoning.
Moreover, $\rho_{s_i}\vDash P_{s_i}$ implies that $\rho_{s_i}=P_{s_i}\rho_{s_i}P_{s_i}$.
Due to \Cref{lem:interproduct}, we have
\begin{align*}
\sum_{i=1}^m\tr( A_{s_i}\rho_{s_i})=\sum_{i=1}^m\tr (A_{s_i}P_{s_i}\rho_{s_i}P_{s_i})=\sum_{i=1}^m\tr (P_{s_i}A_{s_i}P_{s_i}\rho_{s_i}),
\end{align*}
We observe that $\tr(M_{\mathscr{A}}\rho) = \tr(M_{\mathscr{A}'}\rho)$ with $\mathscr{A}' = (P_{s_1}A_{s_1}P_{s_1}, \dots, P_{s_m}A_{s_m}P_{s_m})$. Then 
\[
\{\mathscr{A}' \mid \mathcal{P}\}\,\mathbf{skip}\,\{\mathscr{A} \mid \mathcal{P}\} 
\quad\text{and}\quad 
\{\mathscr{A} \mid \mathcal{P}\}\,\mathbf{skip}\,\{\mathscr{A}' \mid \mathcal{P}\}.
\]
Thus, the qualitative insights from quantum abstract interpretation refine the reasoning process by yielding ``stronger'' postconditions and ``weaker'' preconditions.\footnote{Here, ``stronger'' and ``weaker'' simply refer to the relation
${\tr}(A_{s_i}) \ge {\tr}\big((P_{s_i}P_{s_i})A_{s_i}\big) \ge {\tr}\big(P_{s_i}A_{s_i}P_{s_i}\big)$, as established by \Cref{lem:interproduct}. They do \textbf{not} indicate any ordering relationship between $\mathscr{A}$ and $\mathscr{A}'$.
}

\begin{example}\label{simpleexample_redux}
Consider again the example from \Cref{simpleexample} in which a CNOT gate is applied to the input state $\ket{00}$ (Figure~\ref{fig:CNOT}).
The output state is $\ket{00}$.
Again, we choose \((s_1, s_2) = (\{q_1\}, \{q_2\})\), but now our two-part assertions are able to prove the desired property, now stated as follows:
\begin{align*}
 \vdash\{(\op{0}{0},\op{0}{0})|(\op{0}{0},\op{0}{0})\}\ \mathbf{C}\ \{(\op{0}{0},\op{0}{0})|(\op{0}{0},\op{0}{0})\}.
\end{align*}
\end{example}

\subsection{Reduction}
\label{Se:Reduction}

We now consider the relationship between the correctness $\models^{\rm QAI}$ in QAI and $\models$ in \Cref{quantitative-validity-total}.
%First of all, correctness in QAI
\begin{thm}[Reduction Principle]\label{lifting}
Consider the behavior of a quantum circuit with respect to any input state.
For any projective predicates $\mathcal{P}=(P_{s_1},\cdots, P_{s_m})$ and $\mathcal{Q}=(Q_{s_1},\cdots, Q_{s_m})$,
$\mathcal{P}$ and $\mathcal{Q}$ can be regarded as the observables
$\mathscr{P}:=(P_{s_1},\cdots, P_{s_m})$ and $\mathscr{Q}=(Q_{s_1},\cdots, Q_{s_m})$, respectively.
Then we have the property
\[
\text{If $\models\{\mathscr{P}\}\ \mathbf{C}\ \{\mathscr{Q}\}$ in
the sense of
\Cref{quantitative-validity-total}, then $\models^{\rm QAI} \{\mathcal{P}\}\  \mathbf{C}\ \{\mathcal{Q}\}$.}
\]
% Here, the quantum program is a circuit and can take any input state $\rho$.
\end{thm}

We defer the proof to \Cref{sec:ReductionPrincipleProof}.

\subsection{Logical System with Soundness}
\label{Se:LogicalSystemWithSoundness}

The inference rules for program constructs in \logic are presented in \Cref{fig proof system quanti}. The proof of the following theorem is deferred to \Cref{sec:SoundnessProof}.

\begin{figure}
    \centering
    \resizebox{.8\textwidth}{!}{
    \begin{minipage}{\textwidth}
		\begin{equation*}
        \begin{split}
			&\textsc{Skip}\quad\ \quad\ \frac{}{\{\mathscr{A} \mid \mathcal{P}\}\ \mathbf{Skip}\ \{\mathscr{A} \mid \mathcal{P}\}} \\[0.1cm] 
		&\textsc{Unit-1}\quad\ \frac{  \gamma(\mathcal{P})M_{\mathscr{A}}\gamma(\mathcal{P})\leq \gamma(\mathcal{P})U_F^{\dag}M_{\mathscr{B}}U_F\gamma(\mathcal{P})}{
			\{\mathscr{A} \mid \mathcal{P}\}\ \bar{q}:=U_F\left[\bar{q}\right] \{\mathscr{B} \mid U_F^{\sharp}(\mathcal{P})\}}\ \ \ \ \ \ \\[0.1cm]
		&\textsc{Unit-2}\quad\ \frac{\{\mathscr{A}_i \mid \mathcal{P}_i\}\ \bar{q}:=U_F\left[\bar{q}\right]\{\mathscr{B}_i \mid U_F^{\sharp}(\mathcal{P}_i)\}}{
			\{\oplus_{i}\mathscr{A}_i \mid \mathcal{P}\}\ \bar{q}:=U_F \left[\bar{q}\right]\{\oplus_{i}\mathscr{B}_i \mid \oplus_{i}U_F^{\sharp}(\mathcal{P}_i)\}}\ \ \ \ \ \ \\[0.1cm]
		&\textsc{Seq}\quad\quad\ \
		\frac{\{\mathscr{A} \mid \mathcal{P}\}\ \mathbf{C}_1\ \{\mathscr{D} \mid \mathcal{R}\}\ \ \ \ \ \   \{\mathscr{D} \mid \mathcal{R}\}\ \mathbf{C}_2\ \{\mathscr{B} \mid \mathcal{Q}\}}{\{\mathscr{A} \mid \mathcal{P}\}\ \mathbf{C}_1;\mathbf{C}_2\ \{\mathscr{B} \mid \mathcal{Q}\}}\ \ \ \ \ \ \\[0.1cm]
				& \textsc{Con}\quad\quad\
		\frac{\{\mathscr{A} \mid \mathcal{P}\}\ \mathbf{C}\ \{\mathscr{B} \mid \mathcal{Q}\}, \ \ \mathscr{D}\sqsubseteq\mathscr{A},\ \mathscr{B}\sqsubseteq\mathscr{E},\  \mathcal{R}\sqsubseteq \mathcal{P},\ \mathcal{Q}\sqsubseteq \mathcal{T}}{\{\mathscr{D} \mid \mathcal{R}\}\ \mathbf{C}\ \{\mathscr{E} \mid \mathcal{T}\}}
		\end{split}
        \end{equation*}
    \end{minipage}
    }
	\caption{
        Inference rules for program constructs
        in \logic.
        We can use the proof rules for both forward reasoning or backward reasoning. The \textsc{Skip}, \textsc{Seq}, and \textsc{Con} rules are standard, and operate in conjunction with the ordering on predicates. 
        However, a direct application of the \textsc{Unit-1} rule would break scalability, because it requires computing and applying $\gamma(\mathcal{P})$, i.e., 
        the full concretization of $\mathcal{P}$.
        \textsc{Unit-2} builds on \textsc{Unit-1} to enable scalable reasoning by partitioning the index set $\{1, \dots, m\}$ into disjoint subsets, each corresponding to a group of local predicates to which \textsc{Unit-1} is applied. 
        For tuples of matrices $\mathscr{X}_1, \dots, \mathscr{X}_k$, where $\mathscr{X}_i = (X_{i,1}, \dots, X_{i,m_i})$, we define their concatenation by $\bigoplus_{i=1}^k \mathscr{X}_i := (X_{1,1}, \dots, X_{1,m_1}, \dots, X_{k,1}, \dots, X_{k,m_k})$, which produces a single tuple containing all matrices in order; this notation applies equally to $\mathscr{A}_i$, $\mathscr{B}_i$, and $\mathcal{P}_i$. See \Cref{methods} for details.
    }
    %\\ In the proof system, we only consider $\phi\sqsubseteq \mathcal{P}$ and $\psi\sqsubseteq \mathcal{Q}$ in any correctness formula $\{\phi|\mathcal{P}\} \mathbf{C}\{\psi|\mathcal{Q}\}$.}
	\label{fig proof system quanti}
\end{figure}

	\begin{theorem}[Soundness]
		\label{thm sound QSL}
The proof system in
\Cref{fig proof system quanti} is sound. That is, for quantum program $\mathbf{C}$, $\vdash\{\mathscr{A} \mid \mathcal{P}\}\ \mathbf{C}\ \{\mathscr{B} \mid \mathcal{Q}\}\ {\rm implies}\
    \models\{\mathscr{A} \mid \mathcal{P}\}\ \mathbf{C}\ \{\mathscr{B} \mid \mathcal{Q}\}.$
	\end{theorem}

\subsection{\textsc{Unit-2}: Scalable Application of the \textsc{Unit-1} Rule}
\label{methods}

Among the proof rules in Figure~\ref{fig proof system quanti}, only the rule \textsc{Unit-1} poses a significant challenge to scalability. Specifically, the condition
\begin{align}
  \label{target-giant-inequality}
  \gamma(\mathcal{P}) M_{\mathscr{A}} \gamma(\mathcal{P}) \leq \gamma(\mathcal{P}) \sem{\mathbf{C}}^*(M_{\mathscr{B}}) \gamma(\mathcal{P})
\end{align}
requires computing the projector $\gamma(\mathcal{P})$, which becomes intractable for systems with many qubits. That is, given a postcondition $\mathscr{B}$ and a set of projective predicates $\mathcal{P}$, we are not aware of any scalable method for synthesizing a precondition $\mathscr{A}$ 
that satisfies this inequality.

%\twr{We need some explicit mention of the \textsc{Unit-2} rule in the text.}

To address this issue, we introduce a \emph{compositional approximation} strategy by replacing the \textsc{Unit-1} rule with \textsc{Unit-2}.
The key idea behind \textsc{Unit-2} is to construct inequalities over high-dimensional systems by composing inequalities over smaller subsystems. These approximations trade some precision for tractability, but preserve overall soundness and enable reasoning about large quantum programs.

We first partition the index set $\{{s_1}, \dots, {s_m}\}$ into disjoint subsets $T_1, \dots, T_k$, and consider the following \textsc{Unit-1}-like inequalities for each $T_i$:
\begin{equation}
  \label{Eq:InequalityOne}
  \forall j. \quad
  \gamma(\mathcal{P}) \left( \sum_{i \in T_j} A_{s_i} \otimes I_{[n] \setminus s_i} \right) \gamma(\mathcal{P})
  \leq
  \gamma(\mathcal{P}) U^\dagger \left( \sum_{i \in T_j} B_{s_i} \otimes I_{[n] \setminus s_i} \right) U \gamma(\mathcal{P}).
\end{equation}
Even if each $T_i$ is small, computing $\gamma(\mathcal{P})$ remains challenging. The only aspect that does not fit the \textsc{Unit-2} rule is the role of $\gamma(\mathcal{P})$, which will be discussed next.

\paragraph{Warm-up: Ignoring \boldmath{$\gamma(\mathcal{P})$}}
As a first step, for the sake of scalability we can use the following version of the \textsc{Unit-1} rule, which omits the occurrences of the projection $\gamma(\mathcal{P})$:
\begin{equation}
  \label{Eq:RuleOnePrime}
  \textsc{Unit-1'}\quad\ \frac{  M_{\mathscr{A}} \leq U_F^{\dag}M_{\mathscr{B}}U_F}{
			\{\mathscr{A} \mid \mathcal{P}\}\ \bar{q}:=U_F\left[\bar{q}\right] \{\mathscr{B} \mid U_F^{\sharp}(\mathcal{P})\}}
\end{equation}
Although rule \labelcref{Eq:RuleOnePrime} is imprecise (see \Cref{sec:FirstAttemptAtDefiningJudgments}), it may still yield useful results, as we will see in \Cref{sec:ReasoningAboutTheGHZCircuit}.
The rule \labelcref{Eq:RuleOnePrime} is sound by the unitary rule of QHL~\cite{Ying11}, and it implies the \textsc{Unit-1} rule via \Cref{lem:tr_monotonic_in_Loewner_order}.
\iffalse
This approximation gives rise to:
\begin{equation}
  \label{constraintnoqai}
  M_{\mathscr{A}} \leq \sem{\mathbf{C}}^*(M_{\mathscr{B}})
  \quad \Longleftrightarrow \quad
  \sum_i A_{s_i} \otimes I_{[n] \setminus s_i}
  \leq
  U^\dagger \left( \sum_i B_{s_i} \otimes I_{[n] \setminus s_i} \right) U.
\end{equation}
\Cref{lem:tr_monotonic_in_Loewner_order}
To improve scalability, we partition the index set $\{1, \dots, m\}$ into disjoint subsets $T_1, \dots, T_k$, where each $T_j$ corresponds to a group of local predicates. \fi 

We make use of \textsc{Unit-1'} by structuring a proof into arguments that involve only small matrices, following the pattern on the left-hand side of \Cref{Eq:WarmupInequality}, which, by \Cref{Eq:WarmupInequality}, establishes the premise of the \textsc{Unit-1'} rule.

\begin{equation}
  \label{Eq:WarmupInequality}
  \forall j. \quad
  \sum_{i \in T_j} A_{s_i} \otimes I_{[n] \setminus s_i}
  \leq
  U_F^\dagger \left( \sum_{i \in T_j} B_{s_i} \otimes I_{[n] \setminus s_i} \right) U_F
  \quad \Longrightarrow \quad
  M_{\mathscr{A}} \leq U_F^{\dag}M_{\mathscr{B}}U_F.
\end{equation}
The forward form of this inequality is:
\begin{equation}
  \label{Eq:WarmupInequalityForward}
  \forall j. \quad
  U_F \left( \sum_{i \in T_j} A_{s_i} \otimes I_{[n] \setminus s_i} \right) U_F^\dagger
  \leq
  \sum_{i \in T_j} B_{s_i} \otimes I_{[n] \setminus s_i}.
\end{equation}
Because each inequality involves a small subsystem, constraints can be efficiently discharged using symbolic solvers or semidefinite programming, depending on the structure of $A_{s_i}$ and $B_{s_i}$.

\paragraph{A More Precise Approach: Guarded Inequalities}

The warm-up approach may lose critical precision due to ignoring $\gamma(\mathcal{P})$. To refine this, we incorporate the projection information in a compositional way. Let $\mathcal{P} = \{P_{s_1}, \dots, P_{s_m}\}$, and again partition the index set into disjoint subsets $T_1, \dots, T_k$. For each $T_j$, we define a local over-approximation of the global projector:
\begin{equation}
  \label{Eq:GammaPSubsetPj}
  \gamma(\mathcal{P}) \subseteq  P_j := \bigcap_{i \in T_j} P_{s_i} \otimes I_{[n] \setminus s_i}.
\end{equation}
We can write guarded versions of the local inequalities:
\begin{equation}
  \label{Eq:GuardedInequalityOne}
  \forall j. \quad
  P_j \left( \sum_{i \in T_j} A_{s_i} \otimes I_{[n] \setminus s_i} \right) P_j
  \leq
  P_j U_F^\dagger \left( \sum_{i \in T_j} B_{s_i} \otimes I_{[n] \setminus s_i} \right) U_F P_j.
\end{equation}
These inequalities are more precise than those in \eqref{Eq:WarmupInequality}, because they account for the structure imposed by $\mathcal{P}$, while still avoiding the need to compute the full projector $\gamma(\mathcal{P})$.

\begin{wrapfigure}{r}{0.4\textwidth}
  \centering
  \vspace{-3.0ex}
  \includegraphics[width=0.38\textwidth]{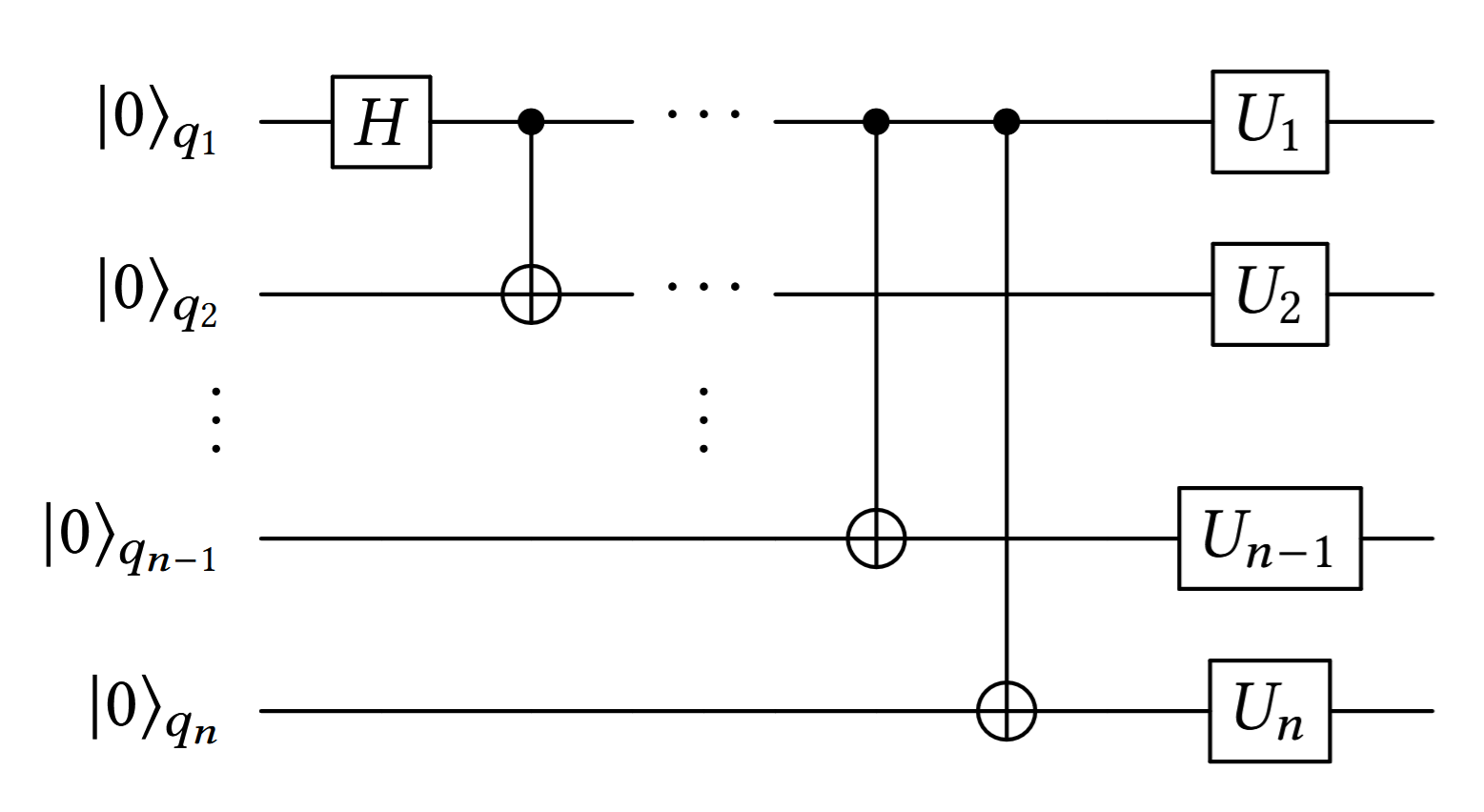}
  \vspace{-1.0ex}
  \caption{GHZ circuit with unitaries \(U_i\).}
  \label{fig:GHZ}
  \vspace{-2.0ex}
\end{wrapfigure}

\paragraph{Summary of Reasoning Tools}
We have practical tools for scalable reasoning under the \textsc{Unit-2} rule:
\begin{itemize}
  \item \textbf{Unguarded inequalities:}
  \Cref{Eq:WarmupInequality,Eq:WarmupInequalityForward} enable fast reasoning, but may be imprecise.
  \item \textbf{Guarded inequalities:}
  \Cref{Eq:GuardedInequalityOne} preserves more of $\gamma(\mathcal{P})$’s structure while remaining tractable.
\end{itemize}
Finally, in \Cref{sec:Eq20ImpliesEq15}, we prove that the guarded inequalities imply the condition stated as \Cref{target-giant-inequality}, thereby validating the correctness of our compositional approximation framework.

\textbf{Remark}: The methodology presented in this section is general. For a given circuit, useful predicates can typically be constructed straightforwardly because each gate in a quantum circuit is represented by a local unitary operator that typically acts on at most two qubits.
In each reasoning step, we analyze a single constant-qubit unitary---typically a one- or
two-qubit gate---and extract the predicate $B_{s_i}$ over its qubits, leaving the remainder untouched.
This approach is sound because the unitary does not affect the reduced density matrices of qubits outside its support. The convenience and effectiveness of this approach are illustrated in the examples throughout the paper.

%% file: GHZ.tex
\section{Quantitative Reasoning about a Generalized GHZ Circuit}
\label{Se:GHZ}

In this section, we consider a generalized GHZ circuit in which half the gates are arbitrary single-qubit unitaries, producing a highly entangled, densely parameterized state beyond the reach of classical simulation methods built on the Gottesman–Knill theorem~\cite{GottesmanKnill}.
The scalable techniques illustrated herein perform
compositional and projection-free reasoning over subsystems, thereby supporting tractable approximation of the circuit’s complex output state.

After presenting \logic in action on the generalized GHZ circuit, we discuss approaches to handling imprecision at the end of \Cref{PrecisionandMitigation}.

\subsection{Reasoning about the GHZ Circuit (\Cref{fig:GHZ})}
\label{sec:ReasoningAboutTheGHZCircuit}

The analysis first applies the ``warm-up'' method from \Cref{methods} (\Cref{Eq:WarmupInequality}).
The second phase of the analysis leverages QAI, employing \Cref{assertion} to confine the output state within a two-dimensional subspace. This critical dimensionality reduction facilitates precise amplitude approximations within our quantitative-reasoning framework. Strikingly, the resulting characterization matches the exact quantum state up to a phase, underscoring that with the right predicates \logic can be used to obtain quite precise results, even though \logic has restricted power because of the concern that proofs be scalable---i.e., both the matrices involved in the proof and the logical derivation grow only polynomially in the number of qubits.

We select the domain
$(\{1,2\},\{2,3\},\cdots,\{n-1,n\})$ for reasoning about \Cref{fig:GHZ}.

% \twr{The presentation in the rest of this section doesn't seem to be connected to \Cref{{methods}}.
% \begin{itemize}
%   \item
%     I don't see where there is a partition imposed on the above domain
%   \item
%     I'm not seeing where \Cref{Eq:WarmupInequality,Eq:InequalityOne} come into play
% \end{itemize}
% }

STEP 1. We first select the precondition to be $\{\mathscr{A} \mid \mathcal{P}\}$
\begin{align*}
\mathscr{A} &= (A_{1,2}, A_{2,3}, \cdots, A_{n-1,n}) = (\op{++}{++}, \op{++}{++}, \cdots, \op{++}{++}) & \\
\mathcal{P} &= (P_{1,2}, P_{2,3}, \cdots, P_{n-1,n}) = (\op{00}{00}, \op{00}{00}, \cdots, \op{00}{00}). &
\end{align*}
One can verify that the initial state $\op{0\cdots 0}{0\cdots 0}\vDash \mathcal{P}$.
We now use our proof rules to compute a postcondition for the GHZ circuit, 
given the precondition $\{\mathscr{A} \mid \mathcal{P} \}$.

After the first $H$ gate, we use inequality \labelcref{{Eq:WarmupInequality}} to derive the following:
\begin{align*}
&\sum_{i}{A_{s_i}}\otimes I_{[n]\setminus \{i,i+1\}}\leq H^{\dag}\sum_{i}{B_{i,i+1}}\otimes I_{[n]\setminus \{i,i+1\}}H\\
%\Longleftrightarrow &H\sum_{i}{A_{s_i}}\otimes I_{[n]\setminus s_i} H^{\dag}\leq\sum_{i}{B_{s_i}}\otimes I_{[n]\setminus s_i}\\
%\Longleftrightarrow &\sum_{i} H{A_{s_i}}\otimes I_{[n]\setminus s_i} H^{\dag}\leq\sum_{i}{B_{s_i}}\otimes I_{[n]\setminus s_i}\\
\Longleftrightarrow~&\sum_{i} H{A_{i,i+1}}\otimes I_{[n]\setminus \{i,i+1\}} H^{\dag}\leq\sum_{i}{B_{i,i+1}}\otimes I_{[n]\setminus \{i,i+1\}}\\
\Longleftrightarrow~&H{A_{1,2}}H^{\dag}\otimes I_{[n]\setminus \{1,2\}}+\sum_{i>1}{A_{i,i+1}}\otimes  HI_{[n]\setminus \{i,i+1\}} H^{\dag}\leq\sum_{i}{B_{i,i+1}}\otimes I_{[n]\setminus \{i,i+1\}}\\
\Longleftrightarrow~&H{A_{1,2}}H^{\dag}\otimes I_{[n]\setminus \{1,2\}}+\sum_{i>1}{A_{i,i+1}}\otimes  I_{[n]\setminus \{i,i+1\}} \leq\sum_{i}{B_{i,i+1}}\otimes I_{[n]\setminus \{i,i+1\}}.
\end{align*}
The final inequality above matches inequality \labelcref{Eq:WarmupInequalityForward} on a term-by-term basis.
We can satisfy inequality \labelcref{Eq:WarmupInequalityForward} when each of the following single-term inequalities hold:
\begin{align*}
  H{A_{1,2}}H^{\dag}\otimes I_{[n]\setminus \{1,2\}}\leq B_{1,2}\otimes I_{[n]\setminus \{1,2\}}
  \qquad\quad
  A_{i,i+1}\otimes  I_{[n]\setminus \{i,i+1\}} \leq B_{i,i+1}\otimes I_{[n]\setminus \{i,i+1\}} \ \mathrm{for}\ i>1.
\end{align*}
Each inequality can be realized as an equality, yielding the postcondition $\{\mathscr{B}\mid\mathcal{Q}\}$,
\begin{align*}
\mathscr{B} &= (B_{1,2}, B_{2,3}, \dots, B_{n-1,n}) 
= (H A_{1,2} H^\dag, A_{2,3}, \dots, A_{n-1,n}) = (\ketbra{0+}{0+}, \ketbra{++}{++}, \dots, \ketbra{++}{++}) \\
\mathcal{Q} &= (Q_{1,2}, Q_{2,3}, \dots, Q_{n-1,n}) 
= (H P_{1,2} H^\dag, P_{2,3}, \dots, P_{n-1,n}) = (\ketbra{+0}{+0}, \ketbra{00}{00}, \dots, \ketbra{00}{00}).
\end{align*}
After the first $\text{CNOT}$ gate has been applied to qubits $q_{1}q_2$, we use the \textsc{Unit} Rule to obtain a postcondition $\{ \mathscr{C} \mid \mathcal{R}\} := \{ (C_{1,2},C_{2,3},C_{3,4}\cdots,C_{n-1,n}) \mid  (R_{1,2},R_{2,3},R_{3,4}\cdots,R_{n-1,n}) \}$.
\iffalse
\begin{align*}
%&(C_{1,2},C_{2,3},C_{3,4}\cdots,C_{n-1,n})=(C_{1,2},C_{2,3},\op{++}{++},\cdots,\op{++}{++})\\
&(R_{1,2},R_{2,3},R_{3,4}\cdots,R_{n-1,n})
=(R_{1,2},R_{2,3},\op{00}{00},\cdots,\op{00}{00})
\end{align*}
\fi
We use QAI to compute $R_{1,2}=\op{00}{00}+\op{11}{11}$, $R_{2,3}=\op{00}{00}+\op{10}{10}$, and $R_{i,i+1}=\op{00}{00}$ for $i > 2$.

What is left to determine are suitable values for $C_{1,2}$ and $C_{2,3}$.
By the \textsc{Unit-2} rule, partition
$\mathscr{C} = (C_{1,2}, C_{2,3}, C_{3,4}, \dots, C_{n-1,n})$ into
$(C_{1,2}, C_{2,3}), (C_{3,4}), \dots, (C_{n-1,n})$. That is,
\begin{small}
\begin{align*}
&\text{CNOT}_{1,2}(I_{[n]\setminus \{i,i+1\}}\otimes B_{i,i+1})\text{CNOT}_{1,2}^{\dag}\leq I_{[n]\setminus \{i,i+1\}}\otimes C_{i,i+1} \ \ \forall \ i>2,\\
&\text{CNOT}_{1,2}( B_{1,2}\otimes I_3+I_1\otimes B_{2,3})\text{CNOT}_{1,2}^{\dag}\leq C_{1,2}\otimes I_3+I_1\otimes C_{2,3}
\end{align*}
\end{small}
We can choose $C_{i,i+1}$ to be $\op{++}{++}$ for $i>2$, and derive the following:
\begin{small}
\begin{align*}
&\text{CNOT}_{1,2}( B_{1,2}\otimes I_3+I_1\otimes B_{2,3})\text{CNOT}_{1,2}^{\dag}\leq C_{1,2}\otimes I_3+I_1\otimes C_{2,3}\\
\Longleftrightarrow & \text{CNOT}_{1,2} B_{1,2}\text{CNOT}_{1,2}^{\dag}\otimes I_3+ \text{CNOT}_{1,2} I_1\otimes B_{2,3}\text{CNOT}_{1,2}^{\dag}\leq C_{1,2}\otimes I_3+I_1\otimes C_{2,3}\\
\Longleftrightarrow & \text{CNOT}_{1,2} \op{0+}{0+}\text{CNOT}_{1,2}^{\dag}\otimes I_3+ \text{CNOT}_{1,2} I_1\otimes \op{++}{++}\text{CNOT}_{1,2}^{\dag}\leq C_{1,2}\otimes I_3+I_1\otimes C_{2,3}\\
\Longleftrightarrow & \op{0+}{0+}\otimes I_3+ \text{CNOT}_{1,2} (\op{0}{0}+\op{1}{1})\otimes \op{++}{++}\text{CNOT}_{1,2}^{\dag}\leq C_{1,2}\otimes I_3+I_1\otimes C_{2,3}\\
\Longleftrightarrow & \op{0+}{0+}\otimes I_3+ \op{0}{0}\otimes \op{++}{++}+ \text{CNOT}_{1,2} \op{1}{1}\otimes \op{++}{++}\text{CNOT}_{1,2}^{\dag}\leq C_{1,2}\otimes I_3+I_1\otimes C_{2,3}\\
\Longleftrightarrow & \op{0+}{0+}\otimes I_3+ \op{0}{0}\otimes \op{++}{++}+ \text{CNOT}_{1,2} \op{1}{1}\otimes \op{+}{+}\text{CNOT}_{1,2}^{\dag} \otimes\op{+}{+}\leq C_{1,2}\otimes I_3+I_1\otimes C_{2,3}\\
 \Longleftrightarrow & \op{0+}{0+}\otimes I_3+I_1\otimes \op{++}{++}\leq C_{1,2}\otimes I_3+I_1\otimes C_{2,3}
 \end{align*}
 \end{small}
where, in the third-to-last and last steps, we use the following facts:
\begin{align}\label{smallinvaraint}
  \text{CNOT}_{1,2}\ket{0}\ket{+}=\ket{0}\ket{+},\ \ \ \text{CNOT}_{1,2}\ket{1}\ket{+}=\ket{1}X\ket{+}=\ket{1}\ket{+}
\end{align}
%\twr{I don't see where $\ket{1}\ket{+}$ played a role anywhere in the above derivation.}
Therefore, we find that the post-state $\mathscr{C}$ predicate is
\begin{equation}
  \label{Eq:FinalCPredicate}
  (C_{1,2},C_{2,3},C_{3,4}\cdots,C_{n-1,n})
  =
  (\op{0+}{0+},\op{++}{++},\cdots,\op{++}{++}).
\end{equation}
(coinciding with the pre-state predicate $\mathscr{B} = (B_{1,2}, B_{2,3}, \dots, B_{n-1,n})$), and the post-state predicate is
\[
  \mathcal{R} = (R_{1,2},R_{2,3},\cdots,R_{n-1,n})
  =
  (\op{00}{00}+\op{11}{11},\op{00}{00}+\op{10}{10},\op{00}{00},\cdots,\op{00}{00}).  
\]

\Cref{Eq:FinalCPredicate} illustrates an advantage of our choice of predicates.
Because of properties such as those given in \Cref{smallinvaraint},
\Cref{Eq:WarmupInequalityForward} remains invariant under the application of $\text{CNOT}$ gates.
This invariance allows us to derive the strongest postcondition,
while preserving the local structure of the matrix representation of predicates.
As a result, we were able to make choices that made the inequalities that we worked with tight (or saturated, i.e., satified as equalities), making it easier to determine the postcondition. The right-hand side of \Cref{Eq:FinalCPredicate} continues to serve as the predicate of local observables, as reasoning continues about the remaining $\text{CNOT}$ gates.

After applying the $\text{CNOT}$ gate on $q_1$ and $q_r$, the postcondition can be chosen as $(\mathscr{D} \mid \mathcal{S})$, where
\begin{align*}
  \mathscr{D}&=(\op{0+}{0+},\op{++}{++},\cdots,\op{++}{++})\\
  \mathcal{S}&= (\op{00}{00}+\op{11}{11},\cdots,\op{00}{00}+\op{11}{11},\op{00}{00}+\op{10}{10},\op{00}{00},\cdots,\op{00}{00}).
\end{align*}
At the right end of the circuit, after the application of  
\( U_1 \otimes U_2 \otimes \cdots \otimes U_n \)—where each \( U_i \) is a single-qubit unitary—the locality structure of the predicates remains unchanged.  
Therefore, we can choose the postcondition to be \( (\mathscr{F} \mid \mathcal{T}) \), where 
\begin{align*}
\mathscr{F}&=(\beta_1\otimes\beta_2,\cdots,\beta_{n-1}\otimes\beta_n)\\
\mathcal{T}&= (\psi_1\otimes \psi_2+\phi_1\otimes\phi_2,\cdots, \psi_{n-1}\otimes \psi_n+\phi_{n-1}\otimes\phi_n)
\end{align*}
with $\beta_{i}=\op{\beta_i}{\beta_i},\ \psi_i=\op{\psi_i}{\psi_i},\ \phi_i=\op{\phi_i}{\phi_i}$, and
\begin{alignat*}{6}
    \ket{\beta_1} &= U_1\ket{0} & \quad \ket{\beta_i} &= U_{i}\ket{+} & \quad \forall i &> 1
  & \qquad\qquad
    \ket{\psi_i}  &= U_i\ket{0} & \quad \ket{\phi_i}  &= U_i\ket{1}   & \quad\forall i &\geq 1.
\end{alignat*}

Let the output state be $\rho = \op{\Psi}{\Psi}$. By \Cref{quantitative-validity-total}, our proof of $\mathscr{F}$ implies that
\begin{align} \label{ghz1}
\frac{n-1}{4}=\sum_{i=1}^{n-1}\tr(\op{00}{00}\op{++}{++})\leq\sum_{i=1}^{n-1}\tr[\rho_{i,i+1}(\beta_i\otimes\beta_{i+1})].
\end{align}
STEP 2.
This step mirrors Step 1, starting from the precondition $((\op{--}{--}, \op{--}{--}, \dots, \op{--}{--}) \mid \mathcal{P})$. The postcondition of the full circuit can then be computed as
\[
\Big((\delta_1 \otimes \delta_2, \dots, \delta_{n-1} \otimes \delta_n) \;\big|\; \mathcal{T}\Big), \quad \mathrm{with} \quad 
\delta_i = \op{\delta_i}{\delta_i}, \;\; \ket{\delta_1} = U_1 \ket{1}, \;\; \ket{\delta_i} = U_i \ket{-} \;\; \forall i>1
\]
These conditions imply that 
\begin{align}\label{ghz2}
\frac{n-1}{4}\leq \sum_{i=1}^{n-1}\tr[\rho_{i,i+1}(\delta_i\otimes\delta_{i+1})].
\end{align}
STEP 3 (QAI influence on QHL). According to Theorem~\ref{assertion}, 
$\mathcal{Q}$ implies that there exist complex numbers \(a\) and \(b\) such that \(|a|^2 + |b|^2 = 1\), and the output state of the GHZ circuit is
\begin{align*}
\ket{\Psi} = a \ket{\psi_1 \cdots \psi_n} + b \ket{\phi_1 \cdots \phi_n}.
\end{align*}

According to $\ip{\psi_i}{\phi_i}=0$, we have
$\rho_{i,i+1}=|a|^2 \psi_i\otimes \psi_{i+1}+|b|^2 \phi_i\otimes \phi_{i+1}.$
Then
\begin{align*}
\tr[(\psi_1\otimes \psi_{2})(\beta_1\otimes\beta_{2})]=\frac{1}{2},\  \tr[(\psi_i\otimes \psi_{i+1})(\beta_i\otimes\beta_{i+1})]=\frac{1}{4} \ \ \forall i>1\\
\tr[(\phi_1\otimes \phi_{2})(\beta_1\otimes\beta_{2})]=0,\ 
\tr[(\phi_i\otimes \phi_{i+1})(\beta_i\otimes\beta_{i+1})]=\frac{1}{4}\ \ \forall i>1\\
\tr[(\psi_1\otimes \psi_{2})(\delta_1\otimes\delta_{2})]=0,\ 
\tr[(\psi_i\otimes \psi_{i+1})(\delta_i\otimes\delta_{i+1})]=\frac{1}{4}\ \  \forall i>1\\
\tr[(\phi_1\otimes \phi_{2})(\delta_1\otimes\delta_{2})]=\frac{1}{2},\ 
\tr[(\phi_i\otimes \phi_{i+1})(\delta_i\otimes\delta_{i+1})]=\frac{1}{4}\ \  \forall i>1.
\end{align*}

\Cref{ghz1,ghz2} imply
\begin{align*}
\frac{n-1}{4}\leq |a|^2/2+\sum_{i=2}^{n-1}\frac{|a|^2+|b|^2}{4}=|a|^2/2+\frac{n-2}{4} \Longrightarrow \frac{1}{2}\leq |a|^2\\
\frac{n-1}{4}\leq |b|^2/2+\sum_{i=2}^{n-1}\frac{|a|^2+|b|^2}{4}=|b|^2/2++\frac{n-2}{4} \Longrightarrow \frac{1}{2}\leq |b|^2\\
\end{align*}
Together with $|a|^2+|b|^2=1$, we have $|a|^2=|b|^2=\frac{1}{2}$. In other words, there exists $\theta$ such that
\begin{align}\label{eqghz}
\ket{\Psi}=\frac{1}{\sqrt{2}}(\ket{\psi_1\cdots \psi_n}+e^{i\theta}\ket{\phi_1\cdots \phi_n}).
\end{align}

\Cref{eqghz} represents a closed-form expression for the circuit's output with a single unknown real parameter---a quantitative result unachievable by qualitative methods like QAI alone.

Moreover, the reasoning process scales with the number of qubits in the circuit:
during the forward-reasoning process described above, for the reasoning steps carried out for each gate, the total size of the matrices that represent local observables and local projections in the pre- and post-conditions is always linear in the number of qubits.
There are $O(n)$ gates;
hence, the total amount of space needed to write down the \logic proof is $O(n^2)$.

\subsection{Precision, Imprecision, and Principles for Selection of Local Observables} \label{PrecisionandMitigation}

\textbf{Inherent Imprecision:}
\logic characterizes the joint behavior of local reduced density matrices in a multi-qubit system, but it cannot capture global correlations that depend on relative phases. Because local reduced matrices cannot encode relative-phase information, any reasoning based solely on them necessarily cannot distinguish states that differ only in relative phase. For example, in the GHZ circuit, even if $U_i = I$, the states $\frac{1}{\sqrt{2}}(\ket{0\cdots 0} \pm \ket{1\cdots 1})$ have identical $k$-qubit reduced density matrices for all $k<n$, namely $\rho_k = \frac{1}{2}(\op{0\cdots 0}{0\cdots 0} + \op{1\cdots 1}{1\cdots 1})$, and are thus indistinguishable by \logic, even though they differ in their \emph{relative-phase information $\theta$ in \Cref{eqghz}}. In other words, while the local properties of each subsystem are identical, the relative-phase information that determines interference and correlations across all qubits is lost. As discussed in the Introduction, the QAI-like and QHL-like components in \logic function as two components of a reduced product, in the sense of abstract interpretation. Consequently, \logic is inherently unable to express relative-phase information.

\textbf{Principles for Selection of Local Observables:} 
An important aspect in applying \logic is the systematic selection of local observables and the corresponding “local domains” relevant to proving a given property of a circuit. For instance, the domain can be chosen according to the circuit structure: as we show in \Cref{sec:AlternativeGHZ}, selecting the sets $\{(1,2),(1,3),(1,4),\dots\}$ as the abstract domain for the GHZ circuit of \Cref{fig:GHZ} yields the same precision as \Cref{eqghz}.

Given a global property of a quantum circuit, incorporating it into our framework requires constructing a tuple of local observables that can capture or closely approximate this property. This challenge touches on a fundamental theme in physics—the relationship between local and global properties—which has been extensively studied and for which many powerful techniques exist. Our approach can also draw inspiration from programming logic, such as best-effort decision procedures in separation logic~\cite{Berdine2005,Calcagno2009} and shape-analysis abstractions~\cite{Sagiv2002,Gopan2004,Lev-Ami2000}, which rely on local information to establish global properties and can guide the principled selection of appropriate local observables.

From these insights, we can state some guiding principles for selecting effective local domains based on circuit structure and entanglement patterns:
(i) Balance precision and tractability while aligning with the properties of interest; in particular, approximate global properties using suitable local observables.  
(ii) Account for the circuit structure:
for circuits consisting of relatively few two-qubit gates, consider all two-qubit pairs appearing in the gate pattern.  
(iii) Refine the domain when it is insufficiently precise for the reasoning task.
We hope that future work will be able to refine and sharpen these principles.

\textbf{Mitigation:} 
For a given circuit, each proof in \logic generates an inequality involving the tuple of selected reduced density matrices.
For example, in the derivation given in \Cref{sec:ReasoningAboutTheGHZCircuit}, Step 1 derives \Cref{ghz1}, which only gives $|a|^2 > \frac{1}{2}$ and is insufficient to obtain \Cref{eqghz}.
By creating a second \logic proof with respect to a different precondition, and combining the two inequalities with the QAI component of the postcondition allowed us to derive $|a|^2 = \frac{1}{2}$. When reasoning with only one predicate, some properties may be lost or become indistinguishable due to the limited information it captures. By combining two or more predicates, we can cross-validate and recover information that would otherwise be missed, improving precision and reducing the “failure” cases that arise from single-predicate reasoning. Nevertheless, there exist properties that cannot be captured solely by tuples of local reduced density matrices; in such cases, \logic alone is insufficient, and integrating it with some non-local assertions as preconditions---such as efficiently representable non-local stabilizers compatible with local reasoning---may be necessary.

%% file: QPE.tex
\section{Quantum Phase Estimation}
\label{Se:qpe}

Quantum Phase Estimation (QPE) is a fundamental quantum algorithm that, given a unitary $U$ and an eigenstate $\ket{\psi}$, estimates the eigenphase $\phi$, where $U \ket{\psi} = e^{2 \pi i \phi}\ket{\psi}$.
The algorithm uses a register of ancilla qubits (whose size determines precision), applies controlled-$U$ operations, and then takes an inverse QFT to concentrate amplitude on the binary fraction closest to $\phi$.
The success probability of QPE is the probability that the measurement outcome is the best approximation to the true eigenphase representable in the ancilla register.
For a single run, this probability is bounded below by \(4/\pi^2 \approx 0.405\), and it increases with additional qubits or repeated executions, ensuring a high likelihood of correct best phase estimation.
We apply \logic to analyze the QPE algorithm, whose intricate entanglement and delicate success probabilities have historically resisted scalable formal verification.
The presentation has two stages: first, we establish the correctness of the Quantum Fourier Transform (QFT) via QAI; second, we leverage this foundation to analyze the full QPE circuit using \logic, deriving a lower bound on its success probability.
To the best of our knowledge, this argument constitutes the first scalable formal verification of QPE and QFT, offering a compositional, tractable methodology for reasoning about both correctness and probabilistic performance in key quantum algorithms.

\subsection{Quantum Fourier Transform (QFT)}

The QFT---the quantum analogue of the discrete Fourier transform
(see \Cref{fig:QFT})---forms the computational core of landmark algorithms such as Shor’s factoring and QPE. We introduce a lossless local-reasoning framework for the QFT grounded in QAI.
Our method derives an abstract output state whose concretization exactly coincides with the true QFT output for all computational-basis inputs.
Moreover, \logic provides proofs that are polynomial in the size of the program while achieving exact semantic correspondence with the true QFT output---and, by extension, for QPE.

\begin{figure}[tb!]
\centering
\vspace{-2pt}
\includegraphics[width=125mm,scale=0.7]{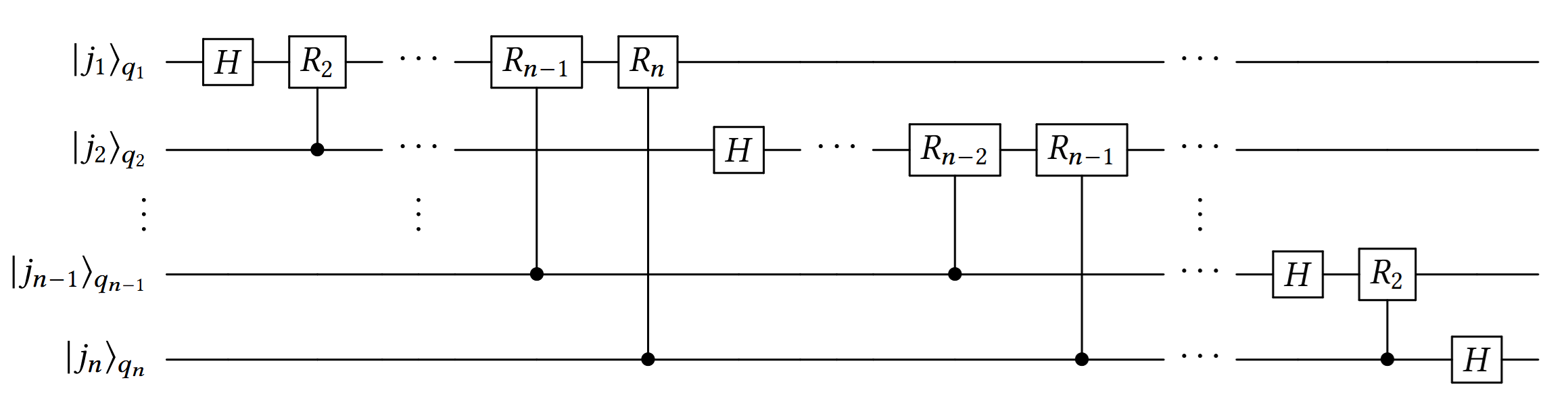}
\vspace{-2pt}
\caption{Quantum Fourier Transform: swap gates
that reverse the qubit order at the circuit’s end are omitted.}
\label{fig:QFT}
\vspace{-3pt}
\end{figure}

The circuit uses Hadamard and $R_m$ gates (see \Cref{sec:quantumbackground});
for \(m > 2\), \(R_m\) lies outside the Clifford group.
We use $\psi=\op{\psi}{\psi}$ for pure state $\ket{\psi}$,$0.x_1x_2\cdots x_n=\sum_{i=1}^n\frac{x_i}{2^i}$, and the following notation:
\begin{align}\label{notation}
  \ket{\psi_{x}}:=\frac{1}{\sqrt{2}}(\ket{0}+e^{2\pi i 0.x}\ket{1}).
\end{align}
We choose the domain \( (\{1\}, \{2\}, \ldots, \{n\}) \).  
For the input state \( \ket{j_1}_{q_1} \otimes \ket{j_2}_{q_2} \otimes \cdots \otimes \ket{j_{n-1}}_{q_{n-1}} \otimes \ket{j_n}_{q_n} \),  
we set the precondition to be
$\mathcal{P}=(P_{1},P_{2},\cdots,P_{n})=(\op{j_1}{j_1},\op{j_2}{j_2},\cdots,\op{j_n}{j_n})$.
In \Cref{sec:QFTviaQAI}, we prove the following property using QAI:
\begin{align}\label{Eq:QFTviaQAI}
\vDash^{QAI} \{\mathcal{P}\}\  QFT\ \{
(\psi_{j_n},\psi_{j_{n-1}j_n},\cdots,\psi_{j_2\cdots j_n},\psi_{j_1\cdots j_n})\}
\end{align}

The postcondition derived from QAI is an abstract state represented as a tuple of density matrices corresponding to pure quantum states.  
By applying the concretization function from \Cref{projectivepredicates-concrete}, we infer that the concrete state lies in the subspace  
\[
\psi_{j_n} \otimes \psi_{j_{n-1}j_n} \otimes \cdots \otimes \psi_{j_2\cdots j_n} \otimes \psi_{j_1\cdots j_n}.
\]  
Because the space so defined is a 1-dimensional subspace, it follows that the density matrix of the output state must exactly equal the pure-state projection onto this vector---that is, $\psi_{j_n} \otimes \psi_{j_{n-1}j_n} \otimes \cdots \otimes \psi_{j_2\cdots j_n} \otimes \psi_{j_1\cdots j_n}$.

During the forward-reasoning process described above, for the reasoning steps carried out for each gate, the total size of the matrices that represent local observables and local projections in the pre- and post-conditions is always linear in the number of qubits.
There are $O(n^2)$ gates;
hence, the total amount of space needed to write down the proof is $O(n^3)$.

%This result achieves the highest possible precision for reasoning QFT, even with exponential space and time simulations.

\subsection{Quantum Phase Estimation (QPE)}

In this section, we present a quantitative local analysis of the QPE algorithm by combining both backward and forward reasoning techniques. We begin by decomposing the QPE algorithm into three constituent circuits and apply \logic to reason about each component individually---see \Cref{fig:QPE}.
For an unknown phase $\theta$ and a given constant $k$, \logic can formally establish that—with probability at least $\frac{4}{\pi^2}$—the QPE algorithm produces an output bitstring whose last $k$ bits constitute the optimal $k$-bit approximation to the least-significant $k$ bits of any $n$-bit binary representation/approximation of $\theta$.

\begin{figure}[tb!]
\centering
\vspace{-2pt}
\includegraphics[width=105mm,scale=0.5]{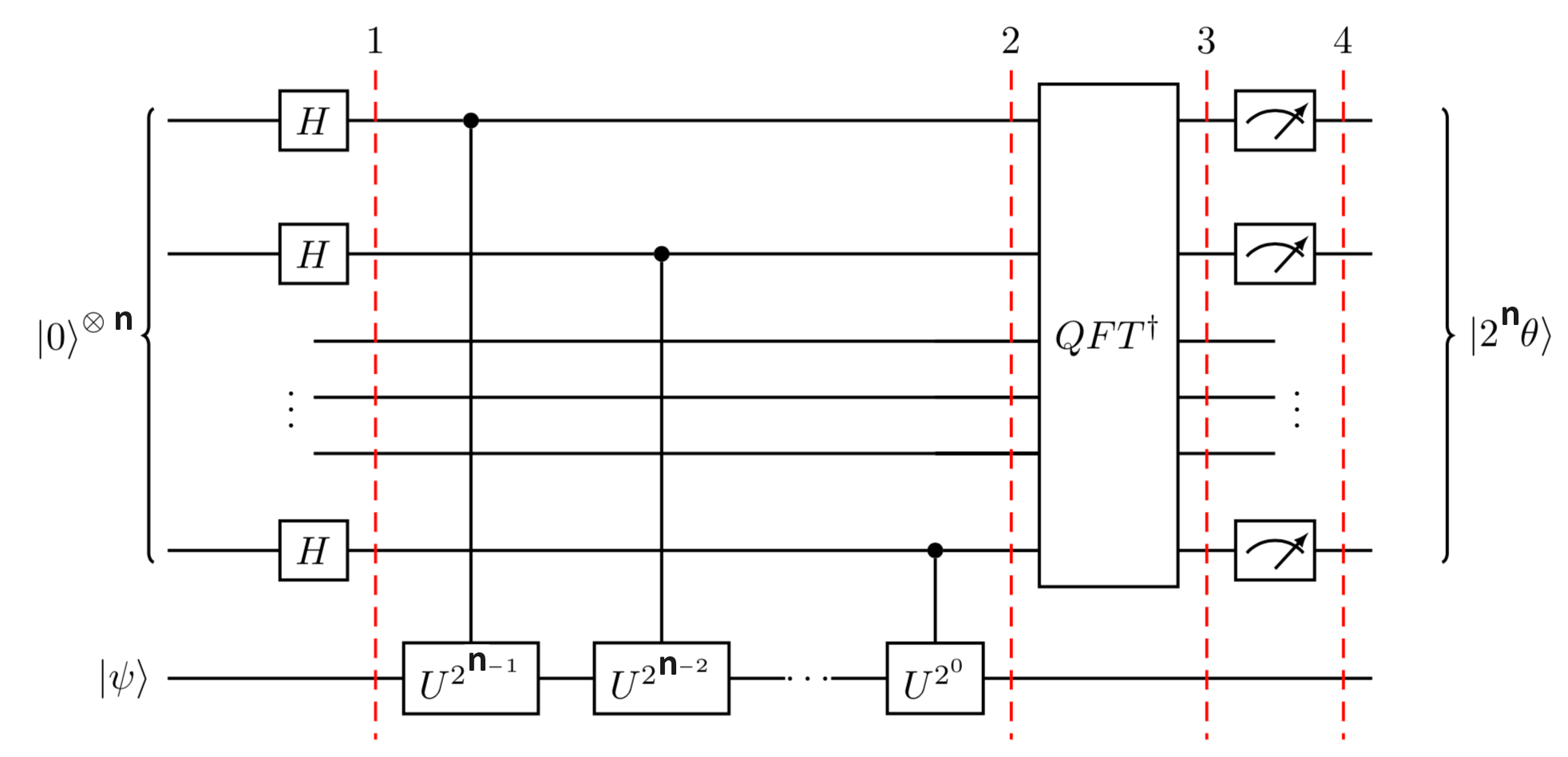}
\vspace{-3pt}
\caption{Quantum Phase Estimation: $U\ket{\psi}=e^{i\theta}\ket{\psi}$. We only consider the circuit without the measurements.}
\label{fig:QPE}
\vspace{-1pt}
\end{figure}

We denote by $\mathbf{C_1}$ the segment of the program that precedes the application of the inverse Quantum Fourier Transform, $\mathcal{QFT}^{\dag}$. The $\mathcal{QFT}^{\dag}$ operation can be decomposed into two parts: the initial sequence of swap gates that reverses the order of the qubits, and the subsequent controlled rotation gates that implement the core of the inverse Fourier transform. We use $\mathbf{C_2}$ to refer specifically to the subcircuit following the swap gates within $\mathcal{QFT}^{\dag}$.

Our analysis focuses on the least significant $k$ qubits, and proceeds in three steps. We will rely on the notation introduced in \Cref{notation}, and apply the QAI-based reasoning framework introduced in \Cref{methods} to verify the behavior of $\mathbf{C_1}$.

\textbf{Step 1: Reasoning about $\mathbf{C_2}$.} Let $U$ be the circuit corresponding to $\mathbf{C_2}$. For $\mathbf{C_2}$, we choose the domain of the predicate as the first $k$ qubits of the first $n$ qubits, together with the last $m$ qubits—the qubits that $U$ acts upon.
In \Cref{sec:ProofOfC1}, we use backward reasoning to show
\begin{align}\label{c-1}
\{\mathscr{A} \mid \mathcal{I}\}\ \mathbf{C_2}\ \{(\op{ j_{n-k+1} \cdots j_n}{ j_{n-k+1} \cdots j_n}) \mid \mathcal{I}\},
\end{align}
where $\mathscr{A}=(\op{\tau}{\tau}\otimes \op{\psi}{\psi})$,
$\ket{\tau} = \ket{\psi_{j_n}}\otimes \cdots \otimes \ket{\psi_{j_{n-k+1}\cdots j_n}}$, and
$M_{\mathscr{A}}=I_{1,\cdots,n-k}\otimes \op{\tau}{\tau}\otimes \op{\psi}{\psi}$.

\textbf{Remark}: Matrix multiplication is efficient for constant dimension.
Here, the matrices are $s^k \times 2^k$, which is a constant independent of $n$.

\textbf{Step 2: Reasoning about SWAP gates.}
The precondition of \Cref{c-1}, $\{\mathscr{A} \mid \mathcal{I}\}$, becomes the postconditon of the SWAP gates.
Because the action of the SWAP gates only changes the last $k$ qubits into the first $k$ qubits, we can write the precondition of the SWAP gates as $\{\mathscr{A}' \mid \mathcal{I}\}$ with $\mathscr{A}'=(\op{\tau}{\tau}\otimes \op{\psi}{\psi})$ on the first $k$ qubits and the last $m$ qubits. These arguments prove
\begin{align}\label{c-2}
  \{\mathscr{A}' \mid \mathcal{I}\}\ \mathbf{S}\ \{\mathscr{A} \mid \mathcal{I}\}
\end{align}
where we use $\mathbf{S}$ to denote the SWAP gates.\footnote{
  We change the abstract domain for simplicity of presentation. This change does not affect our statement's correctness because the SWAP gates' action is clear. If we want to fix the abstract domain, we can consider $(s_1,\cdots,s_m)$ with $s_1$ being the last $k$ qubits together with the last $m$ qubits; $s_2$ being the result of applying the first SWAP gate on $s_1$; $\cdots$; $s_m$ being the result of applying the last SWAP gate on $s_{m-1}$.
}

\textbf{Step 3: Reasoning about $\mathbf{C_1}$.}
For \(\mathbf{C_1}\), we select the predicate domain to include the first \(k\) qubits of the initial \(n\) and the last \(m\) qubits where \(U\) acts.
In \Cref{sec:ProofOfC3}, we prove the triple
\begin{align}\label{c-3}
\{r\op{0\cdots 0}{0\cdots 0}\otimes \op{\psi}{\psi}) \mid (\op{0\cdots 0}{0\cdots 0}\otimes \op{\psi}{\psi}\}\ \mathbf{C_1}\ \{\mathscr{A}' \mid \mathcal{P}\}
\end{align}
where 
\begin{align*}
r&=\Pi_{t=1}^k \cos^2[(2^{n-t}\theta-0. j_{n-t+1}\cdots j_n)\pi]=\frac{\sin^2(2^{n}\theta\pi)}{4^{k}\sin^2[(2^{n-k}\theta-0.j_{n-k+1}\cdots j_n)\pi]}\\
\mathcal{P}&=(\op{\omega}{\omega}\otimes \op{\psi}{\psi})\\
\ket{\omega}&=\frac{1}{2^{k/2}}(\ket{0}+e^{2\pi i 2^{n-1}\theta}\ket{1})\otimes \cdots\otimes (\ket{0}+e^{2\pi i 2^{n-k} \theta}\ket{1}).
\end{align*}

%\Cref{c-2} follows from the \textsf(Unit) Rule. We prove $r\geq \frac{4}{\pi^2}$ in Appendix C.
\iffalse
\begin{align*}
r&=\Pi_{t=1}^k \cos^2[(2^{n-t}\theta-0.j_{n-t+1}\cdots j_n)\pi]
  =\Pi_{t=1}^k \frac{\sin^2[2(2^{n-t}\theta-0.j_{n-t+1}\cdots j_n)\pi]}{4\sin^2[(2^{n-t}\theta-0.j_{n-t+1}\cdots j_n)\pi]}\\
&=\Pi_{t=1}^k \frac{\sin^2[2^{n-t+1}\theta-0.j_{n-t+2}\cdots j_n)\pi]}{4\sin^2[(2^{n-t}\theta-0.j_{n-t+1}\cdots j_n)\pi]}
 = \frac{\sin^2(2^{n}\theta\pi)}{4^{k}\sin^2[(2^{n-k}\theta-0.j_{n-k+1}\cdots j_n)\pi]}.
\end{align*}
\fi
Let $U\ket{\psi}=e^{i\theta}\ket{\psi}$, and $\theta=\frac{a}{2^n}+\epsilon$ with $-\frac{1}{2^{n+1}}\leq \epsilon\leq \frac{1}{2^{n+1}}$ and $a$ is an integer with binary representation $a_1a_{2}\cdots a_n$. For $ j_{n-k}\cdots j_n=a_{n-k}\cdots a_n$, we have
\begin{align*}
r&=\frac{\sin^2(2^{n}\theta\pi)}{4^{k}\sin^2[(2^{n-k}\theta-0.j_{n-k+1}\cdots j_n)\pi]}
  =\frac{\sin^2(2^{n}\epsilon\pi)}{4^{k}\sin^2(2^{n-k}\epsilon\pi)} 
  \geq \frac{|2\cdot 2^{n}\epsilon|^2}{4^{k}\cdot |2^{n-k}\epsilon\cdot\pi|^2}
  \geq \frac{4}{\pi^2}.
 \end{align*}

Together with the \textsf{Seq} Rule and the \textsf{Con} Rule, the proved result can be interpreted as follows:
for any constant \( k \), the last \( k \) bits of the output will—with probability at least \( \frac{4}{\pi^2} \)—match the best \( k \)-bit binary approximation to the least significant \( k \) bits of the phase \( \theta \), provided that the first \( k \) input qubits are initialized to \( \ket{0,\ldots,0} \), regardless of the state of the remaining \( n-k \) qubits.

%% file: related.tex
\section{Related Work}
\label{Se:RelatedWorks}

\textit{Combining Logic and Abstract Interpretation.}
Combining Hoare logic with abstract interpretation is a well-established strategy in classical program analysis.
Many software model-checking tools integrate logic-driven state-space exploration with dataflow information obtained via abstract interpretation, often computed as a preliminary pass. A prominent example is the SeaHorn solver for Constrained Horn Clauses, whose Crab component provides abstract-interpretation capabilities~\cite{Crab,GN21}.
This integration illustrates how abstract interpretation can complement logical reasoning, resulting in verification that is both scalable and precise.
Our work extends this principle to the quantum domain, where the combination supports reasoning about program semantics that would be difficult to achieve with either technique in isolation.

\textit{QHL.}
There are a variety of quantum Hoare logics.
A comparative study in \cite{rand2019verification} examined three representative logics~\cite{CMS,10.1007/978-3-642-10622-4_7,Ying11}. Quantum Hoare logics can be broadly categorized into expectation-based and satisfaction-based approaches. Following the seminal work of D’Hondt and Panangaden~\cite{DP06}, expectation-based approaches~\cite{Ying11,BHY19,li2019quantum,feng2021quantum} use positive operators as assertions for quantum states, defining the expectation that a state $\rho$ satisfies an assertion $M$ as $\tr(M\rho)$. {{In contrast, satisfaction-based logics---exemplified by Zhou et al.~\cite{ZYY19} and indepedently by Unruh~\cite{Unruh,Unruh2}---treat subspaces of the Hilbert space as assertions, providing a direct semantic correspondence between program correctness and the geometry of quantum states.}} Here, a state $\rho$ satisfies an assertion $P$ if the support of $\rho$ is included in $P$. In all of these works, predicates $M$ or $P$ are represented as $2^n \times 2^n$ matrices, which precludes efficient computation and limits scalability for larger quantum systems.

\textit{Quantum Separation Logic.} Quantum reasoning frameworks have made significant progress along multiple fronts. Quantum separation logic, introduced in \cite{zhou2021quantum}, extends the principles of Bunched Implications \cite{o'hearn_pym_1999} to the quantum setting, enabling local reasoning about quantum programs. The framework in \cite{10.1145/3498697} further incorporates hybrid state spaces by supporting both classical variables and the dynamic allocation and deallocation of quantum qubits. {{However, in both approaches, the separating conjunction is interpreted strictly as a tensor product—corresponding to quantum independence—which substantially limits the expressiveness and applicability of the logic. In addition, because predicates are restricted to projection operators, these systems support only qualitative rather than quantitative reasoning about quantum states and their interactions.}}

\textit{QAI.}
\cite{YP21} presented an approach to quantum abstract interpretation for reasoning about quantum circuits, using the satisfaction-based approach. Other works, such as \cite{Bichsel_2023}, investigate the abstract interpretation of quantum programming using variants of the Gottesman-Knill theorem.

{\textit{View abstraction.}
In classical program verification, \emph{view abstraction}~\cite{Abdulla2013} bears a strong resemblance to QAI: both track sets of small-size abstractions of a system.
This idea of focusing on compact representations was already used in shape-analysis frameworks~\cite{Lev-Ami2000,Sagiv2002}, highlighting a conceptual connection between \logic and classical shape analysis. A key difference is that view abstraction targets parameterized model checking, while \logic currently applies to a given circuit with a fixed number of qubits. Nevertheless, this connection suggests an exciting avenue for future work: if \logic can be mechanized, it may be possible to extend it to parameterized quantum circuits \cite{10.1145/3776712}. One could adopt ideas from view abstraction, such as dynamically detecting cut-off points beyond which the state-space search need not continue, or employing the heuristic that when views of size \(k\) fail, switch to views of size \(k+1\).

\textit{Symbolic abstraction, strongest consequence, and weakest sufficient condition.}
The inexpressibility issues discussed in \Cref{sec:FirstAttemptAtDefiningJudgments} are a manifestation of the constraints that one faces when working with an ``impoverished'' logic (or logic fragment).
These issues have been studied in the context of abstract interpretation as what is (now) called the \emph{symbolic-abstraction} problem \cite{DBLP:conf/vmcai/RepsSY04,DBLP:conf/cav/ThakurR12}, and phrased in purely logical terms as the \emph{strongest-consequence} problem \cite[\S5]{DBLP:conf/vmcai/RepsT16}, as follows:
\noindent
\begin{mdframed}[innerleftmargin = 3pt, innerrightmargin = 3pt, skipbelow=-0.5em]
  Given formula $\varphi \in \mathcal{L}$, and another logic $\mathcal{L}'$, find the strongest formula $\psi \in \mathcal{L}'$ such that $\varphi \vDash \psi$.
\end{mdframed}
\noindent
The strongest-consequence problem naturally arises in approximate forwards reasoning, to over-approximate a postcondition.
The discussion in \Cref{sec:FirstAttemptAtDefiningJudgments} concerned backwards reasoning for which one faces the dual problem, the \emph{weakest sufficient-condition} problem:
\noindent
\begin{mdframed}[innerleftmargin = 3pt, innerrightmargin = 3pt, skipbelow=-0.5em]
  Given formula $\varphi \in \mathcal{L}$, and a different logic $\mathcal{L}'$, find the weakest formula $\chi \in \mathcal{L}'$ such that $\chi \vDash \varphi$.
\end{mdframed}

As we saw in \Cref{sec:FirstAttemptAtDefiningJudgments}, the strongest consequence or weakest sufficient condition may not always be expressible in $\mathcal{L}'$, in which case one has to fall back on finding what one hopes is a suitably strong consequence or a suitably weak sufficient condition, respectively.
Scherpelz et al.\ \cite{DBLP:conf/pldi/ScherpelzLC07} presented a best-effort method for computing sufficient conditions as part of an algorithm for creating abstract transformers for use with parameterized predicate abstraction \cite{DBLP:conf/birthday/Cousot03}.
Their method performs weakest liberal precondition (WLP) of a post-state predicate with respect to a concrete transformer $\tau$, and then uses heuristics to identify combinations of pre-state predicates that entail the $\textrm{WLP}$ value.

\textit{Reasoning about Shor's factoring.}
The QFT circuit contains many non-Clifford gates, so variants of the Gottesman-Knill theorem do not directly apply. Shor’s algorithm has been formally verified in~\cite{feng2021quantum,https://doi.org/10.48550/arxiv.2204.07112}. The former uses an expectation-based approach with classical variables, while the latter employs the Coq proof assistant and the Small Quantum Intermediate Representation (sQIR) \cite{Hicks21}. Since Coq operates symbolically, the reasoning about quantum phase estimation in \cite{https://doi.org/10.48550/arxiv.2204.07112} is symbolic.

\textit{Symbolic verification.}
There has been work to extend symbolic-verification techniques to the quantum domain~\cite{10.1145/3725725,CACM:CCLLTY25},
using logic- and automata-based techniques developed for symbolic verification of classical programs to analyze the correctness of quantum programs.
In contrast, \logic is not symbolic in nature.
Reasoning steps in \logic can involve matrix multiplications and other mathematical operations on specific values of specific sizes.

% \logic also does not support parameterized reasoning about families of circuits parameterized on the number of qubits.
% On the contrary, \logic can be used to reason about a specific circuit with a specific number of qubits. 

% we explicitly perform matrix multiplications---even for quantum Fourier transform (QFT) operations. For any given input state $\ket{j_1} \cdots \ket{j_n}$, we compute the corresponding intersection directly.
% \twr{It is not clear what ``computing intersections'' means here.
% If it refers to QAI, it is out of context.
% }

\textit{Packing of variables in abstract domains.}
The motivation for using local observables and local projections defined with respect to a tuple $(s_1,\cdots,s_m)$ of sets of qubit indexes is to make reasoning scalable:
each reasoning step involves only a small number of qubits.
This idea is similar to the idea of ``packing'' variables in numeric domains, as used in Astr{\'e}e \cite{DBLP:conf/pldi/BlanchetCCFMMMR03}:
a program's numeric variables are placed in sets (``packs'') so that each abstract transformer of a numeric abstract domain can operate on a single pack at a time.
As with our qubit sets, each variable can be in multiple packs.

%% file: conclusion.tex
%!TEX root = main.tex
%!TEX spellcheck = en_US
\section{Conclusion and Future work}
\label{sec:conclusion}

This paper introduces a framework for local reasoning about quantum circuits that is both scalable and addresses quantitative properties.
Our logic, \logic, combines QHL, which tracks quantitative properties of quantum states, with QAI, which provides scalable reasoning about qualitative properties.
Unlike a prior scalable framework that is restricted to Clifford circuits \cite{GottesmanKnill}, \logic supports non-Clifford gates, giving it much broader applicability.
By tracking tuples of reduced density matrices via tuples of local observables, \logic fills a long-standing gap in methods for reasoning about quantum programs, enabling scalable and quantitative verification of quantum circuits on classical computers.
We apply \logic to GHZ circuits with non-Clifford unitaries, successfully characterizing the output state up to a relative phase.  
We demonstrate its use on QFT, and derive a lower bound on the success probability of QPE.
Avenues for future research include

% Several avenues for future research emerge from our work:

\begin{enumerate}
    \item \textbf{Other applications.}
    Apply \logic to other circuits, such as W-state circuits and the Bernstein–Vazirani (BV) algorithm, as well as other important classes of quantum algorithms, including parameterized quantum circuits. \logic is well-suited to quantum circuits composed of unitary gates acting on a constant number of qubits. The W-state circuit appears to be a particularly suitable target. A key advantage of the W-state circuit is that the collection of all single-qubit reduced density matrices contains sufficient information to uniquely identify the state within the continuous family of W-type states \cite{Yu2013,Sawicki2013}. 
    The BV circuit requires a model describing the oracle, which remains manageable, but could introduce additional complexity.

    \item \textbf{Quantum programs with measurements and classical control.} Develop the theoretical and technical foundations for handling quantum programs that include measurements, branching, classical control, and loops. 
    Currently, \logic applies to measurement-free circuits, rather than a full programming language with these features, and while we anticipate such extensions are feasible, they introduce additional challenges. In particular, in the presence of loops, verification typically relies on suitable loop invariants, which must be both sufficiently strong to establish the desired property and expressible within the assertion language of \logic. This requirement may rule out certain properties that could otherwise be handled in richer logical frameworks.
    
    \item \textbf{Mechanized proof search.}
    We face several challenges if we are to automate reasoning with \logic. On the theoretical side, a key open problem is how to systematically construct suitable local observables from desired global properties. On the practical side, our approach relies on semidefinite programming (SDP) solvers that can produce corresponding preimages of the optimization outputs, which can serve as candidate predicates or local observables. Given a collection of such candidates, a further challenge is to select high-quality ones effectively.
    
    \item \textbf{Integration with symbolic methods.} Our approach leverages numerical analysis to enable scalable reasoning about general quantum circuits, even in the absence of known algebraic structure. This
    situation is complementary to existing mechanized proof frameworks (e.g., VyZX~\cite{lehmann2026vyzx} and VOQC~\cite{Hicks21}), which rely primarily on symbolic verification. By generating candidate predicates and invariants numerically, our method opens a pathway for integrating numerical and symbolic approaches, expanding the class of quantum programs amenable to mechanized reasoning.
\end{enumerate}

%% file: appendix.tex
%!TEX root = main.tex
%!TEX spellcheck = en_US

\section*{Appendices}

\section{Lemmas and Proofs of Lemmas}
\label{sec:ProofsOfTwoLemmas}

%\twr{You can't just repeat lemmas and have the numbering come out correctly.  There is some way of doing it with  the ``restateable'' mechanism, but I couldn't get it to work correctly.  See https://tex.stackexchange.com/questions/113596/using-a-restatable-before-it-is-stated.}

\begin{lem}\label{inclusion}
For projections $P,Q$,
$P\subseteq Q \Longrightarrow PQ=QP=P.$
\end{lem}

\begin{lem}\label{order}
For projections $P,Q$ and observables $A,B$,
$A\leq B \Longrightarrow PAP\leq PBP.$
\end{lem}

\textbf{\Cref{rdm}.}
Let \( \rho \) be the density matrix of an \( n \)-qubit system, and let \( s \subseteq [n] \). Then for any observable \( A_s \) acting on subsystem \( s \),
\[
\tr\left( (A_s \otimes I_{[n] \setminus s}) \rho \right) = \tr\left( A_s \rho_s \right).
\]

\begin{proof}
We express \( \rho \) as 
\[
\rho = \sum_{i,j} \rho_{i,j} \otimes \ket{i}\bra{j},
\]
where \( \rho_{i,j} \) are operators on subsystem \( s \), and \( \{ \ket{i} \} \) is an orthonormal basis of subsystem \( [n] \setminus s \). Then:
\[
\begin{aligned}
\tr\left( (A_s \otimes I_{[n] \setminus s}) \rho \right)
&= \tr\left( \sum_{i,j} (A_s \rho_{i,j}) \otimes \ket{i}\bra{j} \right) \\
&= \sum_{i} \tr( A_s \rho_{i,i} ) \\
&= \tr\left( A_s \sum_i \rho_{i,i} \right) \\
&= \tr\left( A_s \rho_s \right),
\end{aligned}
\]
where \( \rho_s = \tr_{[n] \setminus s}(\rho) = \sum_i \rho_{i,i} \).

\end{proof}

\section{Proof of Theorem \ref{lifting}}
\label{sec:ReductionPrincipleProof}

\begin{proof} %The semantics $\llbracket S\rrbracket$ of a quantum program $S$ is a quantum operation (or super-operator) $\E$ (\cite{Ying16}, Proposition 3.3.5). We use $\E^\ast$ to denote the dual map of $\E$. Note that for any operator $A$ with $0_\h\sqsubseteq A\sqsubseteq I_\h$, we have: $0_\h\sqsubseteq\E^\ast(A)\sqsubseteq I_\h$, where $0_\h$ and $I_\h$ are the zero and identity operator on the state Hilbert space $\h$.

%{\vskip 3pt}

Assume that $\models\{\mathscr{P}\}C\{\mathscr{Q}\}$
as defined
in \Cref{quantitative-validity-total}. Then for any $\rho$, we have 
\[
\tr(M_{\mathscr{P}}\rho)\leq \tr(M_{\mathscr{Q}}\sem{\mathbf{C}}(\rho))
\]
Let us choose $\rho\models \mathcal{P}$. According to Lemma \ref{lem:satisfyprojection}, we have
\[
\tr(M_{\mathscr{P}}\rho)=\tr (\sum_i P_{s_i}\rho_{s_i})=\sum_i\tr ( P_{s_i}\rho_{s_i})=\sum_i 1=m.
\]
On the other hand,
\begin{align*}
&\tr(M_{\mathscr{Q}}\sem{\mathbf{C}}(\rho))\\
=&\tr \left(\left(\sum_i Q_{s_i}\otimes I_{[n]\setminus s_i}\right)\sem{\mathbf{C}}(\rho)\right)\\
=&\sum_i \tr (Q_{s_i}\tr_{[n]\setminus s_i} \sem{\mathbf{C}}(\rho))\\
\leq& \sum_i \tr (I_{s_i}\tr_{[n]\setminus s_i} \sem{\mathbf{C}}(\rho))\\
=&\sum_i 1=m.
\end{align*}
where the inequality 
follows from \Cref{lem:tr_monotonic_in_Loewner_order}
and $Q_{s_i} \leq I_{s_i}$.

Therefore, $\tr(M_{\mathscr{Q}}\sem{\mathbf{C}}(\rho))$ and $\tr (Q_{s_i}\tr_{[n]\setminus s_i} \sem{\mathbf{C}}(\rho))=1$ for all $i$. According to Lemma \ref{lem:satisfyprojection}, we know \(\tr_{[n]\setminus s_i} \sem{\mathbf{C}}(\rho) \vDash Q_{s_i}\),
which implies
\[
  \models^{\rm QAI} \{\mathcal{P}\}\ C\ \{\mathcal{Q}\},
\]
where $\models_{\rm}^{\rm QAI}\{\mathcal{P}\}\ C\ \{\mathcal{Q}\}$ iff for all $\rho$, $\rho\vDash \mathcal{P}$ implies $\sem{\mathbf{C}}(\rho)\vDash \mathcal{Q}$.
\end{proof}

\section{Proof of Theorem \ref{thm sound QSL}}
\label{sec:SoundnessProof}

\begin{proof} We prove the 
soundness
of each rule.
	
Rule \textsc{Skip}
\begin{align*}
\textsc{Skip}\quad\ \frac{}{\{\mathscr{A} \mid \mathcal{P}\}\ \mathbf{Skip}\ \{\mathscr{A} \mid \mathcal{P}\}}
\end{align*}
For any $\rho\vDash^{QAI} \mathcal{P}$, we have $\sem{\mathbf{Skip}}(\rho)=\rho\vDash^{QAI}\mathcal{P}$. Moreover,
\begin{align*}
\tr[M_{\mathscr{A}}(\rho)]\leq \tr[M_{\mathscr{A}}(\sem{\mathbf{Skip}}(\rho))].
\end{align*}

Rule \textsc{Unit-1}
\begin{align*}
\textsc{Unit-1}\quad\ \frac{  \gamma(\mathcal{P})M_{\mathscr{A}}\gamma(\mathcal{P})\leq \gamma(\mathcal{P})U_F^{\dag}M_{\mathscr{B}}U_F\gamma(\mathcal{P})}{
			\{\mathscr{A} \mid \mathcal{P}\}\ \qbar:=U_F\left[\qbar\right] \{\mathscr{B} \mid U_F^{\sharp}(\mathcal{P})\}}
\end{align*}

For $\rho\vDash^{QAI} \mathcal{P}$, using Theorem \ref{QAI}, we have
\begin{align*}
\sem{\qbar:=U_F\left[\qbar\right]}(\rho)=U_F\rho U_F^{\dag}\vDash^{QAI} U_F^{\sharp}(\mathcal{P}).
\end{align*}
Moreover, we have $\rho\vDash \gamma(\mathcal{P})$;
that is, $\rho = \gamma(\mathcal{P})\rho\gamma(\mathcal{P})$.
According to Lemma \ref{lem:interproduct}, we know
\begin{align*}
{\tr}(M_{\mathscr{A}}\rho)={\tr}(M_{\mathscr{A}}\gamma(\mathcal{P})\rho\gamma(\mathcal{P}))=\tr(\rho \gamma(\mathcal{P})M_{\mathscr{A}}\gamma(\mathcal{P}))
\end{align*}
On the other hand, Lemma \ref{lem:interproduct} also implies 
\begin{align*}
&{\tr}(M_{\mathscr{B}}\sem{\qbar:=U_F\left[\qbar\right]}(\rho))\\
=&{\tr}[M_{\mathscr{B}}\sem{\qbar:=U_F\left[\qbar\right]}(\gamma(\mathcal{P})\rho\gamma(\mathcal{P})]\\
=&\tr[M_{\mathscr{B}}U_F\gamma(\mathcal{P})\rho\gamma(\mathcal{P}) U_F^{\dag}]\\
=&\tr[\rho \gamma(\mathcal{P})U_F^{\dag}M_{\mathscr{B}}U_F\gamma(\mathcal{P})].
\end{align*}
Therefore, the condition 
\begin{align*}
\gamma(\mathcal{P})M_{\mathscr{A}}\gamma(\mathcal{P})\leq \gamma(\mathcal{P})\sem{\mathbf{C}}^*(M_{\mathscr{B}})\gamma(\mathcal{P})
\end{align*}
implies that, for $\rho\vDash^{QAI} \mathcal{P}$, we have
\begin{align*}
{\tr}(M_{\mathscr{A}}\rho)=\tr(\rho \gamma(\mathcal{P})M_{\mathscr{A}}\gamma(\mathcal{P}))\leq\tr[\rho \gamma(\mathcal{P})U_F^{\dag}M_{\mathscr{B}}U_F\gamma(\mathcal{P})]\leq
{\tr}(M_{\mathscr{B}}\sem{\qbar:=U_F\left[\qbar\right]}(\rho))
\end{align*}
by invoking Lemma \ref{lem:interproduct}.

This argument proves Rule \textsc{Unit-1}.

Rule \textsc{Unit-2}
\begin{align*}
\textsc{Unit-2}\quad\ \frac{\{\mathscr{A}_i \mid \mathcal{P}_i\}\ \bar{q}:=U_F\left[\bar{q}\right]\{\mathscr{B}_i \mid U_F^{\sharp}(\mathcal{P}_i)\}}{
			\{\oplus_{i}\mathscr{A}_i \mid \mathcal{P}\}\ \bar{q}:=U_F \left[\bar{q}\right]\{\oplus_{i}\mathscr{B}_i \mid \oplus_{i}U_F^{\sharp}(\mathcal{P}_i)\}}
            \end{align*}

            For $\rho\vDash^{QAI} \mathcal{P}_i$, let $\sigma:=\sem{{\bar{q}:=U_F\left[\bar{q}\right]}}(\rho)$. 
            We observe that 
            \begin{align*}  
            \sum_i M_{\mathscr{A}_i}=M_{\oplus_i\mathscr{A}_i}\\
            \textrm{and}\qquad \sigma\vDash^{QAI} \mathcal{Q}_1,\cdots,\mathcal{Q}_k  \Longrightarrow &\sigma\vDash^{QAI} \oplus_i^k\mathcal{Q}_i
              \end{align*}
         Therefore,
         \begin{align*}   
         &\{\mathscr{A}_i \mid \mathcal{P}_i\}\ \bar{q}:=U_F\left[\bar{q}\right]\{\mathscr{B}_i \mid U_F^{\sharp}(\mathcal{P}_i)\}\\
         \Longrightarrow &\sigma\vDash^{QAI} U_F^{\sharp}(\mathcal{P}_i) \ \land \ \tr[M_{\mathscr{A}_i}\rho]\leq \tr[M_{\mathscr{B}_i}\sigma]\\
         \Longrightarrow &\sigma\vDash^{QAI} U_F^{\sharp}(\mathcal{P}_1) \ \land\  \cdots \ \land\  \sigma\vDash^{QAI} U_F^{\sharp}(\mathcal{P}_k) 
         \ \land\ \tr[ \sum_i M_{\mathscr{A}_i}\rho]\leq \tr[\sum_i M_{\mathscr{B}_i}\sigma]\\
          \Longrightarrow &\sigma\vDash^{QAI} \oplus_i U_F^{\sharp}(\mathcal{P}_1) \ \land \ 
         \tr[M_{\oplus_i\mathscr{A}_i}\rho]\leq \tr[\oplus_i M_{\mathscr{B}_i}\sigma]
         \end{align*}
         
This argument proves Rule \textsc{Unit-2}.

Rule \textsc{Seq}
\begin{align*}
\textsc{Seq}\quad
		\frac{\{\mathscr{A} \mid \mathcal{P}\}\ \mathbf{C}_1\ \{\mathscr{D} \mid \mathcal{R}\}\ \ \ \ \ \   \{\mathscr{D} \mid \mathcal{R}\}\ \mathbf{C}_2\ \{\mathscr{B} \mid \mathcal{Q}\}}{\{\mathscr{A} \mid \mathcal{P}\}\ \mathbf{C}_1;\mathbf{C}_2\ \{\mathscr{B} \mid \mathcal{Q}\}}
\end{align*}

For $\rho\vDash^{QAI} \mathcal{P}$, we have
\begin{align*}
\{\mathscr{A} \mid \mathcal{P}\}\ \mathbf{C}_1\ \{\mathscr{D} \mid \mathcal{R}\}
\Longrightarrow \sem{\mathbf{C}_1}(\rho)\vDash^{QAI} \mathcal{R}, \ \ \ \tr[M_{\mathscr{A}}\rho]\leq \tr[M_{\mathscr{D}}(\sem{\mathbf{C}_1}(\rho))].
\end{align*}
According to $\{\mathscr{D} \mid \mathcal{R}\}\ \mathbf{C}_2\ \{\mathscr{B} \mid \mathcal{Q}\}$ and $\sem{\mathbf{C}_1}(\rho)\vDash^{QAI} \mathcal{R}$, we obtain that
\begin{align*}
\sem{\mathbf{C}_2}(\sem{\mathbf{C}_1}(\rho))\vDash^{QAI} \mathcal{Q}\\
\tr[M_{\mathscr{D}}(\sem{\mathbf{C}_1}(\rho))]\leq \tr[\sem{M_{\mathscr{B}}(\mathbf{C}_2}(\sem{\mathbf{C}_1}(\rho)))]
\end{align*}
According to the fact that $\sem{\mathbf{C}_1;\mathbf{C}_2}(\rho))=\sem{\mathbf{C}_2}(\sem{\mathbf{C}_1}(\rho))$, we know that if $\rho\vDash^{QAI} \mathcal{P}$,
\begin{align*}
\sem{\mathbf{C}_1;\mathbf{C}_2}(\rho))\vDash^{QAI} \mathcal{Q}\\
\tr[M_{\mathscr{A}}\rho]\leq \tr[M_{\mathscr{B}}(\sem{\mathbf{C}_1;\mathbf{C}_2}(\rho))]
\end{align*}
This argument proves Rule \textsc{Seq}.

Rule \textsc{Con}
\begin{align*}
\textsc{Con}\quad
		\frac{\{\mathscr{A} \mid \mathcal{P}\}\mathbf{C}\{\mathscr{B} \mid \mathcal{Q}\}, \ \ \mathscr{D}\sqsubseteq\mathscr{A},\ \mathscr{B}\sqsubseteq\mathscr{E},\  \mathcal{R}\sqsubseteq \mathcal{P},\ \mathcal{Q}\sqsubseteq \mathcal{T}}{\{\mathscr{D} \mid \mathcal{R}\}\ \mathbf{C}\ \{\mathscr{E} \mid \mathcal{T}\}}
\end{align*}

For any $\rho\vDash^{QAI} \mathcal{R}$, we have
\begin{align*}
\rho\vDash^{QAI} \mathcal{R}\sqsubseteq \mathcal{P}.
\end{align*}
Moreover, 
\begin{align*}
\{\mathscr{A} \mid \mathcal{P}\}\ \mathbf{C}\ \{\mathscr{B} \mid \mathcal{Q}\}
\Longrightarrow 
\sem{\mathbf{C}}(\rho)\vDash^{QAI} \mathcal{Q}\sqsubseteq \mathcal{T},\ \ 
\tr[M_{\mathscr{A}}\rho]\leq \tr[M_{\mathscr{B}}(\sem{\mathbf{C}}(\rho))].
\end{align*}
According to Lemma \ref{mono}, we have
\begin{align*}
&M_{\mathscr{D}}\leq M_{\mathscr{A}}, \ \  M_{\mathscr{B}}\leq M_{\mathscr{E}}\\ 
\Longrightarrow &\tr[M_{\mathscr{D}}\rho]\leq \tr[M_{\mathscr{A}}\rho]\leq \tr[M_{\mathscr{B}}(\sem{\mathbf{C}}(\rho))]\leq \tr[M_{\mathscr{E}}(\sem{\mathbf{C}}(\rho))].
\end{align*}
This argument proves the correctness of Rule \textsc{Con}.
\end{proof}

%\section{Appendix C: Bound of $r$ in Reasoning QPE}
%This section provides the complete computation to bound $r$ in Section 5.2.

\section{Equation (\ref{Eq:GuardedInequalityOne}) implies Equation (\ref{target-giant-inequality})}
\label{sec:Eq20ImpliesEq15}

\begin{align}
   & \quad \textrm{for all $j$} & P_j \left( \sum_{i \in T_j} A_{s_i} \otimes I_{[n] \setminus s_i} \right) P_j
   &\leq
    P_j U^\dagger \left( \sum_{i \in T_j} B_{s_i} \otimes I_{[n] \setminus s_i} \right) U P_j   \tag{\ref{Eq:GuardedInequalityOne}} \\
  \label{Eq:InequalityTwo}
\Longleftrightarrow  & \quad \textrm{for all $j$} & \gamma(\mathcal{P}) P_j \left(\sum_{i\in T_j }A_{s_i}\otimes I_{[n]\setminus s_i} \right) P_j \gamma(\mathcal{P})
   &\leq \gamma(\mathcal{P}) P_j U^{\dag} \left(\sum_{i\in T_j }B_{s_i}\otimes I_{[n]\setminus s_i}\right)U P_j \gamma(\mathcal{P}) \\
  \label{Eq:InequalityThree}
  \Longrightarrow & \quad \textrm{for all $j$} &
  \gamma(\mathcal{P}) \left(\sum_{i\in T_j }A_{s_i}\otimes I_{[n]\setminus s_i}\right) \gamma(\mathcal{P})
  &\leq \gamma(\mathcal{P}) U^{\dag} \left(\sum_{i\in T_j }B_{s_i}\otimes I_{[n]\setminus s_i}\right)U \gamma(\mathcal{P}) \\
  \Longrightarrow &&\sum_j\gamma(\mathcal{P}) \left(\sum_{i\in T_j }A_{s_i}\otimes I_{[n]\setminus s_i}\right) \gamma(\mathcal{P})
  &\leq \sum_j\gamma(\mathcal{P}) U^{\dag} \left(\sum_{i\in T_j }B_{s_i}\otimes I_{[n]\setminus s_i}\right)U \gamma(\mathcal{P}) \notag \\
  \label{Eq:SomethingKnownInTheTooLargeSystem}
  \Longrightarrow &&\gamma(\mathcal{P}) \left(\sum_{i}A_{s_i}\otimes I_{[n]\setminus s_i}\right) \gamma(\mathcal{P})
  &\leq \gamma(\mathcal{P}) U^{\dag}\left(\sum_{i}B_{s_i}\otimes I_{[n]\setminus s_i}\right)U \gamma(\mathcal{P}) \\
\Longleftrightarrow &&   \gamma(\mathcal{P}) M_{\mathscr{A}} \gamma(\mathcal{P}) &\leq \gamma(\mathcal{P}) \sem{\mathbf{C}}^*(M_{\mathscr{B}}) \gamma(\mathcal{P}) \tag{\ref{target-giant-inequality}}
\end{align}
where in going from \Cref{Eq:GuardedInequalityOne} to \Cref{Eq:InequalityTwo}, we used \Cref{order};
and in going from \Cref{Eq:InequalityTwo} to \Cref{Eq:InequalityThree}, we used \Cref{Eq:GammaPSubsetPj} and \Cref{inclusion}.

\section{An Alternative Approach to the {Generalized} GHZ Circuit in Figure \ref{fig:GHZ}}
\label{sec:AlternativeGHZ}

As in \Cref{sec:ReasoningAboutTheGHZCircuit}, the analysis first applies the ``warm-up'' method from \Cref{methods} (\Cref{Eq:WarmupInequality}), and the second phase of the analysis leverages QAI.
This time, we select the domain $(\{1,2\},\{1,3\},\cdots,\{1,n\})$ for reasoning about Figure \ref{fig:GHZ}.

STEP 1. We first select the precondition to be $\{\mathscr{A} \mid \mathcal{P}\}$
\begin{align*}
\mathscr{A} &= (A_{1,2}, A_{1,3}, \cdots, A_{1,n}) = (\op{++}{++}, \op{++}{++}, \cdots, \op{++}{++}) & \\
\mathcal{P} &= (P_{1,2}, P_{1,3}, \cdots, P_{1,n}) = (\op{00}{00}, \op{00}{00}, \cdots, \op{00}{00}). &
\end{align*}
One can verify that the initial state $\op{0\cdots 0}{0\cdots 0}\vDash \mathcal{P}$.
We now use our proof rules to compute a postcondition for the GHZ circuit, 
given the precondition $\{\mathscr{A} \mid \mathcal{P} \}$.

By the \textsc{Unit-2} rule, we partition $\mathscr{A} = (A_{1,2}, A_{1,3}, \cdots, A_{1,n})$ into $(A_{1,2}), (A_{1,3}), \cdots, (A_{1,n})$.

After the first $H$ gate, we can use inequality \labelcref{{Eq:WarmupInequality}} to derive the following:
\begin{align*}
&{A_{1,i}}\otimes I_{[n]\setminus \{1,i\}}\leq H^{\dag}{B_{1,i}}\otimes I_{[n]\setminus \{1,i\}}H\\
\Longleftrightarrow &H {A_{1,i}}\otimes I_{[n]\setminus \{1,i\}}H^{\dag} \leq {B_{1,i}}\otimes I_{[n]\setminus \{1,i\}}\\
\Longleftrightarrow &H {A_{1,i}}H^{\dag}\otimes I_{[n]\setminus \{1,i\}} \leq {B_{1,i}}\otimes I_{[n]\setminus \{1,i\}}.
\end{align*}
Each inequality can be realized as an equality, yielding the postcondition $\{\mathscr{B}\mid\mathcal{Q}\}$,
\begin{align*}
\mathscr{B} &= (B_{1,2}, B_{1,3}, \dots, B_{1,n}) 
= (H A_{1,2} H^\dag, HA_{1,3}H^\dag, \dots, HA_{1,n}H^\dag) = (\ketbra{0+}{0+}, \ketbra{0+}{0+}, \dots, \ketbra{0+}{0+})\\
\mathcal{Q} &= (Q_{1,2}, Q_{1,3}, \dots, Q_{1,n}) 
= (H P_{1,2} H^\dag, HP_{2,3}H^\dag, \dots, HP_{n-1,n}H^\dag) = (\ketbra{+0}{+0}, \ketbra{+0}{+0}, \dots, \ketbra{+0}{+0}).
\end{align*}
After the first $\text{CNOT}$ gate has been applied to qubits $q_{1}q_2$, we use QAI to compute $R_{1,2}=\op{00}{00}+\op{11}{11}$, $R_{1,i}=\op{00}{00}+\op{10}{10}$, for $i > 1$.

By the \textsc{Unit-2} rule, we know
\begin{small}
\begin{align*}
&\text{CNOT}_{1,2}(I_{[n]\setminus \{1,i\}}\otimes B_{1,i})\text{CNOT}_{1,2}^{\dag}\leq I_{[n]\setminus \{1,i\}}\otimes C_{1,i}.
\end{align*}
\end{small}
We can choose $C_{1,i}=\op{0+}{0+}$, for $i\geq 2$. 
Therefore, we find that the post-state $\mathscr{C}$ predicate is
\begin{equation*}
  \label{Eq:FinalCPredicateInProof}
  (C_{1,2},C_{1,3},C_{1,4}\cdots,C_{1,n})
  =
  (\op{0+}{0+},\op{0+}{0+},\cdots,\op{0+}{0+}).
\end{equation*}
(coinciding with the pre-state predicate $\mathscr{B} = (B_{1,2}, B_{1,3}, \dots, B_{1,n})$), and the post-state QAI predicate is
\[
  \mathcal{R} = (R_{1,2},R_{1,3},\cdots,R_{1,n})
  =
  (\op{00}{00}+\op{11}{11},\op{00}{00}+\op{10}{10},\op{00}{00}+\op{10}{10},\cdots,\op{00}{00}+\op{10}{10}).  
\]
After applying the $\text{CNOT}$ gate on $q_1$ and $q_r$, the postcondition can be chosen as $(\mathscr{D} \mid \mathcal{S})$, where
\begin{align*}
  \mathscr{D}&=  (\op{0+}{0+},\op{0+}{0+},\cdots,\op{0+}{0+})\\
  \mathcal{S}&= (\op{00}{00}+\op{11}{11},\cdots,\op{00}{00}+\op{11}{11},\op{00}{00}+\op{10}{10},\cdots,\op{00}{00}+\op{10}{10}).
\end{align*}

At the right end of the circuit, after the application of  
\( U_1 \otimes U_2 \otimes \cdots \otimes U_n \)---where each \( U_i \) is a single-qubit unitary---the locality structure of the predicates remains unchanged.  
Therefore, we can choose the postcondition to be \( (\mathscr{F} \mid \mathcal{T}) \), where 
\begin{align*}
\mathscr{F}&=(\beta_1\otimes\beta_2,\cdots,\beta_{1}\otimes\beta_n)\\
\mathcal{T}&= (\psi_1\otimes \psi_2+\phi_1\otimes\phi_2,\cdots, \psi_{1}\otimes \psi_n+\phi_{1}\otimes\phi_n)
\end{align*}
with $\beta_{i}=\op{\beta_i}{\beta_i},\ \psi_i=\op{\psi_i}{\psi_i},\ \phi_i=\op{\phi_i}{\phi_i}$, and
\begin{align*}
\ket{\beta_1}&=U_1\ket{0},\ \  \ket{\beta_i}=U_{i}\ket{+}\ \  \forall i>1\\
\ket{\psi_i}&=U_i\ket{0},\ \ \ \ket{\phi_i}=U_i\ket{1}\ \ \forall i\geq 1.
\end{align*}
This implies
\begin{align} 
\frac{n-1}{4}\leq\sum_{i=2}^{n}\tr[\rho_{1,i}(\beta_1\otimes\beta_{i})].
\end{align}

STEP 2. We first select the precondition to be $\{\mathscr{A} \mid \mathcal{P}\}$
\begin{align*}
\mathscr{A}' &=  (\op{--}{--}, \op{--}{--}, \cdots, \op{--}{--}) & \\
\mathcal{P} &= (P_{1,2}, P_{1,3}, \cdots, P_{1,n}) = (\op{00}{00}, \op{00}{00}, \cdots, \op{00}{00}). &
\end{align*}
The postcondition of the full circuit can then be computed as
\[
\Big((\delta_1 \otimes \delta_2, \dots, \delta_{1} \otimes \delta_n) \;\big|\; \mathcal{T}\Big), \quad \mathrm{with} \quad 
\delta_i = \op{\delta_i}{\delta_i}, \;\; \ket{\delta_1} = U_1 \ket{1}, \;\; \ket{\delta_i} = U_i \ket{-} \;\; \forall i>1
\]
These conditions imply that 
\begin{align}
\frac{n-1}{4}\leq \sum_{i=2}^{n}\tr[\rho_{1,i}(\delta_1\otimes\delta_{i})].
\end{align}
The remainder of the argument is the same as that given in \Cref{sec:ReasoningAboutTheGHZCircuit}.

\section{Reasoning about QFT using QAI (Equation~\ref{Eq:QFTviaQAI})}
\label{sec:QFTviaQAI}

Our goal is to show that the following property holds:
\begin{align*}
  \vDash^{QAI} \{\mathcal{P}\}\  QFT\ \{(\psi_{j_n},\psi_{j_{n-1}j_n},\cdots,\psi_{j_2\cdots j_n},\psi_{j_1\cdots j_n})\}   \tag{\ref{Eq:QFTviaQAI}}.
\end{align*}

The quantum gates used in the circuit are the Hadamard gate and the phase gate $R_m$, $R_m$ does not belong to the Clifford group for $m>2$.
\begin{align*}
{ H={\frac {1}{\sqrt {2}}}{\begin{pmatrix}1&1\\1&-1\end{pmatrix}}\qquad {\text{and}}\qquad R_{m}={\begin{pmatrix}1&0\\0&e^{2\pi i/2^{m}}\end{pmatrix}}}
\end{align*} 
We use $\psi=\op{\psi}{\psi}$ to represent pure state $\ket{\psi}$, $0.x_1x_2\cdots x_n = \sum_{i=1}^n\frac{x_i}{2^i}$, and the following notation 
\begin{align}\label{notation-1}
\ket{\psi_{x}}:=\frac{1}{\sqrt{2}}(\ket{0}+e^{2\pi i 0.x}\ket{1}).
\end{align}
We choose the domain
$(\{1\},\{2\},\cdots,\{n\})$, and the precondition to be
\begin{align*}
\mathcal{P}=(P_{1},P_{2},\cdots,P_{n})=(\op{j_1}{j_1},\op{j_2}{j_2},\cdots,\op{j_n}{j_n}).
\end{align*}

At the beginning of the program, according to $\beta_k=\op{j_k}{j_k}\models_p P_k$, we have the input $$\ket{\beta}=\ket{j_1}_{q_1}\otimes\ket{j_2}_{q_2}\otimes\cdots\otimes\ket{j_{n-1}}_{q_{n-1}}\otimes\ket{j_n}_{q_n} \vDash^{QAI} \mathcal{P}.$$
After applying the first $H$ gate on $q_1$, we compute the postcondition, which becomes
\begin{align*}
(H\op{j_1}{j_1}H,\op{j_2}{j_2},\cdots,\op{j_n}{j_n})=(\psi_{j_1},\op{j_2}{j_2},\cdots,\op{j_n}{j_n})
\end{align*}
where $\ket{\psi_{j_1}}:=\frac{1}{\sqrt{2}}(\ket{0}+e^{2\pi i 0.j_1}\ket{1})$ and $\psi_{j_1}=\op{\psi_{j_1}}{\psi_{j_1}}$,
by notifying $e^{2\pi i 0.j_1}=-1$ if $j_1=0$, otherwise, $e^{2\pi i 0.j_1}=1$.
 
Applying the controlled-$R_2$ gate, we compute
\begin{align*}
P_{1,2}:=\psi_{j_1}\otimes I_2\cap I_1\otimes \op{j_2}{j_2}
=\psi_{j_1}\otimes \op{j_2}{j_2}.
\end{align*}
The postcondition becomes
\begin{align*}
&(\supp(\tr_2 CR_2 A_{1,2} CR_2^{\dag}),\supp(\tr_1 CR_2 A_{1,2} CR_2^{\dag}),\op{j_3}{j_3},\cdots,\op{j_n}{j_n})\\
= &(\psi_{j_1j_2},\op{j_2}{j_2},\cdots,\op{j_n}{j_n})
\end{align*}
We continue applying the controlled-$R_3$, $R_4$ through $R_n$ gates and compute our postcondition, each of which adds an
an extra bit to the phase of the coefficient of the first $\ket{1}$. At the end of this procedure, we have the postcondition
\begin{align*}
(\psi_{j_1j_2\cdots j_n},\op{j_2}{j_2},\cdots,\op{j_n}{j_n})
\end{align*}

Next, we perform a similar procedure on the second qubit. The Hadamard gate puts us
in the postcondition 
\begin{align*}
(\psi_{j_1j_2\cdots j_n},H\op{j_2}{j_2}H,\op{j_3}{j_3},\cdots,\op{j_n}{j_n})=(\psi_{j_1j_2\cdots j_n},\psi_{j_2},\op{j_3}{j_3},\cdots,\op{j_n}{j_n})
\end{align*}
The controlled-$R_2$ through $R_{n-1}$ gates yield the predicate
\begin{align*}
(\psi_{j_1j_2\cdots j_n},\psi_{j_2\cdots j_n},\op{j_3}{j_3},\cdots,\op{j_n}{j_n})
\end{align*}
We continue in this fashion for each qubit, giving a final predicate
\begin{align*}
(\psi_{j_1j_2\cdots j_n},\psi_{j_2\cdots j_n},\psi_{j_3\cdots j_n},\cdots,\psi_{j_n}).
\end{align*}
It follows the Swap operation between the qubit $i$ and the qubit $n+1-i$ for $1\leq i\leq n$.
After applying $SWAP(1,n)$, we obtain 
\begin{align*}
(\psi_{j_n},\psi_{j_2\cdots j_n},\cdots,\psi_{j_{n-1}j_n},\psi_{j_1\cdots j_n})
\end{align*}
by observing
\begin{align*}
\supp(\tr_n SWAP(1,n)[\psi_{j_1\cdots j_n}\otimes I_n]\cap [I_1\otimes \psi_{j_n}] SWAP(1,n) )&=\psi_{j_n}\\
\supp(\tr_1 SWAP(1,n)[\psi_{j_1\cdots j_n}\otimes I_n]\cap [I_1\otimes \psi_{j_n}] SWAP(1,n) )&=\psi_{j_1\cdots j_n}.
\end{align*}
After all the swap operations, the postcondition is
\begin{align*}
(\psi_{j_n},\psi_{j_{n-1}j_n},\cdots,\psi_{j_2\cdots j_n},\psi_{j_1\cdots j_n})
\end{align*}
The postcondition derived from QAI is an abstract state represented as a tuple of density matrices corresponding to pure quantum states.

\section{Proofs of Equations (\ref{c-1}) and (\ref{c-3}) from Section \ref{Se:qpe}}
\label{sec:ProofOfTwoEquations}

\subsection{Proof of \Cref{c-1}}
\label{sec:ProofOfC1}

We wish to show that the following triple holds:
\begin{equation*}
 \{\mathscr{A} \mid \mathcal{I}\}\ \mathbf{C_2}\ \{(\op{ j_{n-k+1} \cdots j_n}{ j_{n-k+1} \cdots j_n}) \mid \mathcal{I}\}, \tag{\ref{c-1}}
\end{equation*}
where $\mathscr{A}=(\op{\tau}{\tau}\otimes \op{\psi}{\psi})$,
$\ket{\tau} = \ket{\psi_{j_n}}\otimes \cdots \otimes \ket{\psi_{j_{n-k+1}\cdots j_n}}$, and
$M_{\mathscr{A}}=I_{1,\cdots,n-k}\otimes \op{\tau}{\tau}\otimes \op{\psi}{\psi}$.
\Cref{c-1} is equivalent to 
\begin{align*}
M_{\mathscr{A}}\leq & ~U^{\dag} [I_{1,\cdots,n-k}\otimes\op{ j_{n-k+1} \cdots j_n}{ j_{n-k+1}\cdots j_n}\otimes \op{\psi}{\psi}]U \\
                  = & ~U^{\dag} (I_{1,\cdots,n-k}\otimes\op{ j_{n-k+1} \cdots j_n}{ j_{n-k+1}\cdots j_n}) U\otimes \op{\psi}{\psi}.
\end{align*}
To see that this inequality holds, we need several observations, starting from
\begin{align*}
  & ~H(1)(I_{1,\cdots,n-k}\otimes\op{ j_{n-k+1} \cdots j_n}{ j_{n-k+1}\cdots j_n}) H(1)^{\dag}\\
= & ~CR_n(n,1) (I_{1,\cdots,n-k}\otimes\op{ j_{n-k+1} \cdots j_n}{ j_{n-k+1}\cdots j_n})CR_2(n,1)^{\dag}\\
= & ~CR_2(2,1) (I_{1,\cdots,n-k}\otimes\op{ j_{n-k+1} \cdots j_n}{ j_{n-k+1}\cdots j_n})CR_2(2,1)^{\dag}\\
= & ~H(n-k+2)(I_{1,\cdots,n-k}\otimes\op{ j_{n-k+1} \cdots j_n}{ j_{n-k+1}\cdots j_n}) H(n-k+2)^{\dag}\\
= & ~I_{1,\cdots,n-k}\otimes\op{ j_{n-k+1} \cdots j_n}{ j_{n-k+1}\cdots j_n}.
\end{align*}
That is, $I_{1,\dots,n-k} \otimes \ket{j_{n-k+1}\cdots j_n}\!\bra{j_{n-k+1}\cdots j_n}$ is invariant under $H(1)$, $CR_n(n,1)$, $CR_2(2,1)$, $\dots$, $H(n-k+2)$.
Let $V = H(1)^{-1}CR_2(2,1)^{-1}CR_{3}(3,1)^{-1}\cdots H(n-k+2)^{-1}$.
We have
\begin{align*}
V^{\dag}(I_{1,\cdots,n-k}\otimes\op{ j_{n-k+1} \cdots j_n}{ j_{n-k+1}\cdots j_n}) V=I_{1,\cdots,n-k}\otimes\op{ j_{n-k+1} \cdots j_n}{ j_{n-k+1}\cdots j_n}
\end{align*}
We let $U=VU_k$ with $U_k$ being the sub-circuit which only applied on the last $k$-qubits, i.e., 
\begin{align*}
U_k = H(n-k+1)^{-1}\cdots H(n-1)^{-1} CR_2(n,n-1)^{-1} H(n)^{-1},
\end{align*}
Then
\begin{align*}
  & ~U^{\dag} (I_{1,\cdots,n-k}\otimes\op{ j_{n-k+1} \cdots j_n}{ j_{n-k+1}\cdots j_n}) U\otimes \op{\psi}{\psi}\\
= & ~U_k^{\dag} (I_{1,\cdots,n-k}\otimes\op{ j_{n-k+1} \cdots j_n}{ j_{n-k+1}\cdots j_n}) U_k\otimes \op{\psi}{\psi}\\
= & ~I_{1,\cdots,n-k}\otimes U_k^{\dag}\op{ j_{n-k+1} \cdots j_n}{ j_{n-k+1}\cdots j_n} U_k\otimes \op{\psi}{\psi}
\end{align*}
$U_k^{\dag}$ is the standard Quantum Fourier Transform applied on input $\op{ j_{n-k+1} \cdots j_n}{ j_{n-k+1}\cdots j_n}$. Performing direct matrix multiplication (or according to the last subsection of reasoning about the Quantum Fourier Transform), we know that
\begin{align*}
 U_k^{\dag} \op{ j_{n-k+1} \cdots j_n}{ j_{n-k+1}\cdots j_n} U_k=\op{\tau}{\tau}
\end{align*}
%This argument proves \Cref{c-1}. 

\subsection{Proof of \Cref{c-3}}
\label{sec:ProofOfC3}

We wish to show that the following triple holds:
\begin{align*}
\{r\op{0\cdots 0}{0\cdots 0}\otimes \op{\psi}{\psi}) \mid (\op{0\cdots 0}{0\cdots 0}\otimes \op{\psi}{\psi}\}\ \mathbf{C_1}\ \{\mathscr{A}' \mid \mathcal{P}\}  \tag{\ref{c-3}}
\end{align*}
where 
\begin{align*}
r&=\Pi_{t=1}^k \cos^2[(2^{n-t}\theta-0. j_{n-t+1}\cdots j_n)\pi]=\frac{\sin^2(2^{n}\theta\pi)}{4^{k}\sin^2[(2^{n-k}\theta-0.j_{n-k+1}\cdots j_n)\pi]}\\
\mathcal{P}&=(\op{\omega}{\omega}\otimes \op{\psi}{\psi})\\
\ket{\omega}&=\frac{1}{2^{k/2}}(\ket{0}+e^{2\pi i 2^{n-1}\theta}\ket{1})\otimes \cdots\otimes (\ket{0}+e^{2\pi i 2^{n-k} \theta}\ket{1}).
\end{align*}

\noindent
We first use QAI \cite{YP21} to prove 
\begin{align}\label{key-2}
\{\cdot \mid (\op{0\cdots 0}{0\cdots 0}\otimes \op{\psi}{\psi})\}\ \mathbf{C_1}\ \{\cdot \mid \mathcal{P}\}.
\end{align}

\noindent
After applying the first $H$ gates, 
the post-condition becomes
\begin{align*}
\{\cdot \mid (\op{+\cdots +}{+\cdots +}\otimes \op{\psi}{\psi}\}.
\end{align*}

For $CU^{2^{n-1}}\cdots CU^{2^{n-k}}$,
direct matrix computation leads to the post-condition 
$\{\cdot \mid \mathcal{P}\}$.

For $CU^{2^{n-k+1}}\cdots CU^{2^{0}}$,
there will be no change to
the observable on the first $k$ qubits.
For example, for the first gate
$CU^{2^{0}}$, which is $CU$,
we have the post-condition $\{\cdot \mid \op{+\cdots +}{+\cdots +} \otimes\op{\psi}{\psi}\}$
by observing
\begin{align*}
  & ~\supp[\tr_n CU(\op{\omega}{\omega}\otimes I_{n}\otimes \op{\psi}{\psi})CU^{\dag}]\\
= & ~\supp[\tr_n CU(\op{\omega}{\omega}\otimes (\op{0}{0}+\op{1}{1})\otimes \op{\psi}{\psi})CU^{\dag}]\\
= & ~\supp[ 2 \op{\omega}{\omega}\otimes \op{\psi}{\psi}]\\
= & ~\op{\omega}{\omega}\otimes \op{\psi}{\psi}.
\end{align*}
Similarly for $CU^{2^{n-k+1}}\cdots CU^{2^{0}}$.
This argument proves \Cref{key-2}.

To prove \Cref{c-3}, we observe
\begin{align*}
&~\ \ \op{0\cdots 0}{0\cdots 0}\otimes I\otimes \op{\psi}{\psi}(r\op{0\cdots 0}{0\cdots 0}\otimes I\otimes \op{\psi}{\psi})\op{0\cdots 0}{0\cdots 0}\otimes I\otimes \op{\psi}{\psi}\\
=&~r\op{0\cdots 0}{0\cdots 0}\otimes I\otimes \op{\psi}{\psi};
\end{align*}
and 
\begin{align*}
&~\ \ \op{0\cdots 0}{0\cdots 0}\otimes I\otimes \op{\psi}{\psi}[U(C_1)^{\dag}(\op{\tau}{\tau}\otimes I\otimes \op{\psi}{\psi})U(C_1)]\op{0\cdots 0}{0\cdots 0}\otimes I\otimes \op{\psi}{\psi}\\
=&~\ \ \op{0\cdots 0}{0\cdots 0}\otimes I\otimes \op{\psi}{\psi}[V(C_1)^{\dag}(\op{\tau}{\tau}\otimes I\otimes \op{\psi}{\psi})V(C_1)]\op{0\cdots 0}{0\cdots 0}\otimes I\otimes \op{\psi}{\psi}\\
=&~x\op{0\cdots 0}{0\cdots 0}\otimes I\otimes \op{\psi}{\psi};
\end{align*}
where 
\begin{align*}
V(C_1)&=CU^{2^{n-k}}\cdots CU^{2^{n-1}}\\
x&=\langle 0\cdots 0\psi|V(C_1)^{\dag}\op{\tau\psi}{\tau \psi}\ket{0\cdots 0\psi}\\
  &=\tr[(\op{\tau\psi}{\tau\psi})V(C_1) \op{0\cdots 0\psi}{0\cdots 0\psi}V(C_1)^{\dag}]\\
  &=\tr[\op{\tau\psi}{\tau\psi}\op{\omega\psi}{\omega\psi}]\\
  &= |\ip{\tau}{\omega}|^2\\
  &=\Pi_{t=1}^k  |\frac{(\langle 0|+e^{-2\pi i 0.j_{n-t+1}\cdots j_n}\langle 1|)(\ket{0}+e^{2\pi i 2^{n-t} \theta}\ket{1}}{2}|^2\\
  &=\Pi_{t=1}^k \cos^2[(2^{n-t}\theta-0.j_{n-t+1}\cdots j_n)\pi]\\
  &=r.
\end{align*}